\documentclass[a4paper]{CSML}
\pdfoutput=1

\usepackage{lastpage}
\newcommand{\dOi}{14(2:9)2018}
\lmcsheading{}{1--\pageref{LastPage}}{}{}%
{Jan.~15,~2016}{May ~15, 2018}{}

\usepackage[bookmarks=true,colorlinks=false]{hyperref}
\hypersetup{hidelinks}
\usepackage{url}

\usepackage[textsize=footnotesize]{todonotes}
\usepackage[draft,footnote,nomargin,marginclue]{fixme}
\usepackage{microtype}

\usepackage{ifthen}
\newboolean{full}
\setboolean{full}{true}

\usepackage{preamble_journal}

\def\boldsymbol#1{#1}

\begin{document}
%
%
\FXRegisterAuthor{sm}{asm}{SM}
\FXRegisterAuthor{hu}{ahu}{HU}
\FXRegisterAuthor{ja}{aja}{JA}

%
%
\title[A Categorical Approach to Syntactic Monoids]{A Categorical Approach to Syntactic Monoids}

\author[J.~Ad\'amek]{Ji\v{r}\'\i\ Ad\'amek{\rsuper{a}}}
\author[S.~Milius]{Stefan Milius{\rsuper{b}}}
\thanks{Stefan Milius acknowledges support by the Deutsche Forschungsgemeinschaft (DFG) under project MI~717/5-1}
\author[H.~Urbat]{Henning Urbat{\rsuper{c}}}
\thanks{Ji\v{r}\'{i} Ad\'amek and Henning Urbat acknowledge
    support by the Deutsche Forschungsgemeinschaft (DFG) under project AD
    187/2-1}

\address{{\lsuper{a}}Institut f\"ur Theoretische Informatik, Technische Universit\"at Braunschweig, Germany}
\email{adamek@iti.cs.tu-bs.de}

\address{{\lsuper{b,c}}Lehrstuhl f\"ur Theoretische Informatik, Friedrich-Alexander-Universit\"at Erlangen-N\"urnberg, Germany}
\email{mail@stefan-milius.eu, henning.urbat@fau.de}


\begin{abstract}
The syntactic monoid of a language is generalized to the level of a symmetric 
monoidal closed category $\D$. This allows for a uniform treatment of several 
notions of syntactic algebras known in the literature, including the syntactic 
monoids of Rabin and Scott ($\D=$ sets), the syntactic ordered monoids of Pin 
($\D =$ posets), the syntactic semirings of Pol\'ak ($\D=$ semilattices), and 
the syntactic associative algebras of Reutenauer ($\D$ = vector spaces).  
Assuming that $\D$ is a commutative variety of algebras or ordered algebras, we prove 
that the 
syntactic $\D$-monoid of a language $L$ can be constructed as a quotient of a 
free $\D$-monoid modulo the syntactic congruence of $L$, and that it is 
isomorphic to the transition $\D$-monoid of the minimal automaton for $L$ in 
$\D$. Furthermore, in the case where the variety $\D$ is locally finite, we 
characterize the regular languages as precisely the languages with finite 
syntactic $\D$-monoids.
\end{abstract}

\maketitle

%
%

\section{Introduction}

One of the successes of the theory of coalgebras is that ideas from automata theory can be developed at a level of abstraction where they apply uniformly to many different types of systems. In fact, classical deterministic automata are a standard example of coalgebras for an endofunctor. And that automata theory can be studied with coalgebraic methods rests on the observation that formal languages form the final coalgebra.

The present paper contributes to a new category-theoretic view of \emph{algebraic} automata theory. In this theory one starts with an elegant machine-independent notion of language recognition: a language $L\seq X^*$ is \emph{recognized} by a monoid morphism $e: X^* \to M$ if it is the preimage under $e$ of some subset of $M$. Regular languages are then characterized as precisely the languages recognized by finite monoids. A key concept, introduced by Rabin and Scott~\cite{rs59} (and earlier in unpublished work of Myhill), is the \emph{syntactic monoid} of a language $L$. It serves as a canonical algebraic recognizer of $L$, namely the smallest $X$-generated monoid recognizing $L$. Two standard methods to construct the syntactic monoid are: 
\begin{enumerate}[label=(\arabic*)]
\item as a quotient of the free monoid $X^*$ modulo the \emph{syntactic congruence} of $L$, which is a two-sided version of the well-known Myhill-Nerode equivalence, and 
\item as the \emph{transition monoid} of the minimal automaton for $L$. 
\end{enumerate}
In addition to syntactic monoids there are several related notions of 
syntactic algebras for (weighted) languages in the literature, most 
prominently the syntactic ordered monoids of Pin \cite{pin95}, the syntactic 
idempotent semirings of Pol\'ak~\cite{polak01} and the syntactic associative 
algebras of Reutenauer~\cite{reu80}, all of which admit constructions similar 
to (1) and (2). A crucial observation is that monoids, ordered monoids, 
idempotent semirings and associative algebras are precisely the monoid objects 
in the categories of sets, posets, semilattices and vector spaces, 
respectively. Moreover, all these categories are symmetric monoidal closed 
w.r.t. their usual tensor product. 

The main goal of our paper is to develop a theory of algebraic recognition in 
a general symmetric monoidal closed category $\D=(\D,\t,I)$, covering all the above syntactic algebras  uniformly. Following Goguen 
\cite{goguen75}, a \emph{language} in $\D$ is a morphism $L: \xs\ra Y$ where 
$X$ is a fixed object of inputs, $Y$ is a fixed object of outputs, and $\xs$ 
denotes the free $\D$-monoid on $X$. And a \emph{$\D$-automaton} is given by 
the picture below: it consists of an object of states $Q$, a morphism $i$ 
representing the initial state, an output morphism $f$, and a transition 
morphism $\delta$ which may be presented in its curried form $\lambda\delta$.
\begin{equation}
  \label{eq:daut}
  \vcenter{
    \xymatrix{
      & X \t Q \ar[d]^\delta &\\
      I \ar[r]^i & Q\ar[r]^f \ar[d]^{\lambda\delta} & Y\\
      & [X,Q] &
    }
  }
\end{equation}
As observed by Rutten \cite{rutten}, this means that an automaton is at the same time an \emph{algebra} $I+X\t Q\xra{[i,\delta]} Q$ for the functor $FQ= I+ X\t Q$, and a \emph{coalgebra} $Q\xra{\langle f,\lambda \delta\rangle } Y\times [X,Q]$ for the functor $TQ = Y\times[X,Q]$. It turns out that much of the  classical (co-)algebraic theory of automata in the category of sets extends to this level of generality. Thus Goguen \cite{goguen75} demonstrated that the initial algebra for $F$ coincides with the free $\D$-monoid $\xs$, and that every language is accepted by a unique minimal $\D$-automaton. We will add to this picture the observation that the final coalgebra for $T$ is carried by the object of languages $[\xs,Y]$, see Proposition \ref{prop:fincoalg}.

In Section \ref{sec:syn} we introduce the central concept of our paper, the 
\emph{syntactic $\D$-monoid} of a language $L:\xs\ra Y$, which by definition 
is the smallest $X$-generated $\D$-monoid recognizing $L$. In Corollary \ref{cor:syntmonexists} we give a natural condition on a monoidal category 
that ensures the existence of a syntactic $\D$-monoid for every 
language $L$. Furthermore, assuming that $\D$ is a commutative variety of 
algebras or ordered algebras, we will show that the above constructions (1) 
and (2) for the classical syntactic monoid adapt to our general setting: the 
syntactic $\D$-monoid is (1) the quotient of $\xs$ modulo the syntactic 
congruence of $L$ (Theorem~\ref{prop:synmon}), and (2) the transition 
$\D$-monoid 
of the minimal $\D$-automaton for $L$ (Theorem~\ref{thm:tran}). As special 
instances we recover syntactic monoids ($\D=$ sets), 
syntactic ordered monoids ($\D=$ posets), syntactic semirings 
($\D=$ semilattices) and syntactic associative algebras ($\D=$ vector spaces). In addition, our 
categorical setting yields 
new types of syntactic algebras ``for free''. For example,  we will identify 
monoids with zero as the algebraic structures representing partial automata 
(the case $\D=$ pointed sets) and obtain the \emph{syntactic monoid with 
zero} for a given language. Similarly, by taking as $\D$ the variety of 
algebras with an involutive unary operation we obtain \emph{syntactic 
involution monoids}.

Most of the results of our paper apply to arbitrary languages. In Section \ref{sec:rat} we will investigate \emph{$\D$-regular languages}, that is, languages accepted by $\D$-automata with a finitely presentable object of states. Under suitable assumptions on $\D$, we will prove that a language is $\D$-regular iff its syntactic $\D$-monoid is carried by a finitely presentable object (Theorem~\ref{thm:reg}). We will also derive a dual characterization of the syntactic $\D$-monoid which is new even in the ``classical'' case $\D=$ sets: if $\D$ is a locally finite variety, i.e. all finitely generated algebras are finite, and if moreover some other locally finite variety $\C$ is dual to $\D$ on the level of finite objects, the syntactic $\D$-monoid of $L$ dualizes to the local variety of languages in $\C$ generated by the reversed language of $L$.

\medskip
\textbf{Related work.} 
This paper is a reworked full version of the conference paper \cite{amu15}. 
Apart from including full proofs, it has two new contributions. First, we 
present in Section \ref{sec:existsynmon} a purely categorical existence 
criterion for syntactic monoids in an abstract symmetric monoidal closed category, whereas in 
\cite{amu15} only a set-theoretic construction was given. Secondly, we treat the 
case of ordered algebras explicitly, covering new examples such as the 
syntactic ordered monoids of Pin \cite{pin95}.

 Another categorical 
approach to (classical) syntactic monoids appears in the work of 
Ballester-Bolinches, Cosme-Llopez and Rutten \cite{bcr14}. These authors 
consider automata in the category of sets specified by \emph{equations} or 
dually by \emph{coequations}, which leads to a construction of the automaton 
underlying the syntactic monoid of a language. The fact that it forms the 
transition monoid of a minimal automaton is also interpreted in that setting. 
In the present paper we take a more general approach by 
studying algebraic recognition in an arbitrary symmetric monoidal closed 
category $\D$. One important source of inspiration for our categorical setting 
was the work of Goguen \cite{goguen75}.

In the recent papers \cite{ammu14,ammu15} we presented a categorical view of \emph{varieties of languages}, another central topic of algebraic automata theory. Building on the duality-based approach of Gehrke, Grigorieff and Pin \cite{ggp08}, we generalized Eilenberg's variety theorem and its local version to the level of an abstract (pre-)duality between algebraic categories. The idea to replace monoids by monoid objects in a commutative variety $\D$ originates in this work.

Another categorical setting for algebraic language theory  can be found in Boj\'anczyk 
\cite{boj15}. He considers, in lieu of monoids in commutative varieties, 
Eilenberg-Moore algebras for an arbitrary monad on sorted sets, and defines 
syntactic congruences in this more general setting. Our Theorem 
\ref{prop:synmon} is a special case of \cite[Theorem 3.1]{boj15}.

\section{Preliminaries}
\label{sec:pre}

Throughout this paper we investigate deterministic automata in a symmetric monoidal closed
category. In all concrete applications, this category will be a commutative 
variety 
of algebras or ordered algebras.  Recall that 
for a finitary signature~$\Sigma$ a \emph{variety of
  $\Sigma$-algebras} is a class of
$\Sigma$-algebras specified by equations $s = t$ between $\Sigma$-terms.  
Equivalently, by
Birkhoff's HSP theorem, it is a class of $\Sigma$-algebras closed
under products, subalgebras and quotients (= homomorphic images). Similarly,
\emph{ordered $\Sigma$-algebras} are posets equipped with order-preserving
$\Sigma$-operations, and their morphisms are order-preserving 
$\Sigma$-homomorphisms. A
\emph{quotient} of an ordered algebra $B$ is represented by a surjective morphism 
$e\colon
B \epito A$, and a \emph{subalgebra} of~$B$ is represented by an 
\emph{order-reflecting}
morphism~$m\colon A
\rightarrowtail B$, i.e.\ $mx \leq
my$ iff $x \leq y$.  A \emph{variety of ordered $\Sigma$-algebras} is a class 
of ordered $\Sigma$-algebras
specified by inequations $s \leq t$ between $\Sigma$-terms. Equivalently, by 
Bloom's HSP
theorem \cite{blo76}, it is a class of ordered $\Sigma$-algebras closed under 
products, 
subalgebras and quotients.

A variety $\D$ of algebras or ordered algebras is called \emph{commutative} (a.k.a. \emph{entropic}) 
if, for any two algebras $A$ and $B$ in $\D$, the set 
$[A,B]$ of all homomorphisms from $A$ to $B$ forms a subalgebra $[A,B]\monoto 
B^{\under{A}}$ of the product of $\under{A}$ copies of $B$. Equivalently, the 
corresponding monad over $\Set$ (resp. $\Pos$, the category of posets) is 
commutative in the sense of Kock \cite{kock70}.

\begin{example}\label{ex:entvar}
In our applications, we shall work with the following commutative varieties:
\begin{enumerate}[label=(\arabic*)]
\item  $\Set$ is a commutative variety (over the empty signature) with $[A,B] = B^A$. 

\item $\Pos$ is a commutative variety of ordered algebras (over the empty signature). Here $[A,B]$ is the poset of monotone functions from $A$ to $B$, ordered pointwise.

\item A \emph{pointed set} $(A,\bot)$ is a set $A$ together with a chosen point $\bot\in A$. The category $\PSet$ of pointed sets and point-preserving functions is a commutative variety. The point of $[(A,\bot_A),(B,\bot_B)]$ is the constant function with value $\bot_B$.

\item An \emph{involution algebra} is a set with an involutive unary operation $x\mapsto \tl x$, i.e.~$\tl{\tl x} = x$.  We call $\tl x$ the \emph{complement} of $x$. Morphisms are functions $f$ preserving complements, i.e. $f(\tl x) = \widetilde{f(x)}$. The variety $\Inv$ of involution algebras is commutative. Indeed, the set $[A,B]$ of all homomorphisms is an involution algebra with pointwise complementation.

\item All other examples we treat in our paper are varieties of modules over a semiring. Recall that a \emph{semiring} $\S=(S,+,\o,0,1)$ consists of a commutative monoid $(S,+,0)$ and a monoid $(S,\o,1)$ subject the following equational laws:
\[ 0s = s0 = 0,\quad r(s+t) = rs + rt\quad\text{and}\quad (r+s)t = rt +st.\]
A \emph{module} over a semiring $\S$ is a commutative monoid $(M,+,0)$ together with a scalar multiplication $\o: S\times M\ra M$ (written simply as juxtaposition $(s,x) \mapsto sx$, as usual) such that the following laws hold:
\[
\begin{array}{r@{\ }c@{\ }l@{\qquad}r@{\ }c@{\ }l@{\qquad}r@{\ }c@{\ }l}
  (r+s)x & = & rx+sx, & r(x+y) & = & rx + ry, & (rs)x & = & r(sx), \\
  0x & = & 0, & 1x & = & 1, & r0 & = & 0.
\end{array}
\]
We denote by $\SMod{\S}$ the category of $\S$-modules and module homomorphisms (i.e.~$\S$-linear maps). This is a commutative variety; here $[A,B]$ is the set of all module homomorphisms with pointwise module structure. Three interesting special cases of $\SMod{\S}$ are:
\begin{enumerate}[label=(\alph*)]
\item $\S=\{0,1\}$, the boolean semiring with $1+1=1$: the category $\JSL$ of join-semilattices with the least element $0$, and homomorphisms preserving joins and $0$;

\item $\S=\Int$: the category $\Ab$ of abelian groups and group homomorphisms;

\item  $\S=\K$ (a field): the category $\Vect{\K}$ of vector spaces over $\K$ and linear maps.
\end{enumerate}
\end{enumerate}
\end{example}

\begin{notation}
For any variety $\D$ of algebras or ordered algebras, we denote by $\Psi: 
\Set\ra\D$ the left adjoint to the forgetful functor 
$\under{\mathord{-}}:\D\ra\Set$. Thus $\Psi X_0$ is the free object of $\D$ on 
the set $X_0$. \takeout{We may assume that $X$ is a subset of $\under{\Psi X}$ 
and the universal map is the inclusion $X\monoto \under{\Psi X}$. (If $\D$ 
contains at least one algebra with more than one elements, then such a choice 
of $\Psi$ is always possible because all universal maps are injective.)}
\end{notation}

\begin{example}
\begin{enumerate}[label=(\arabic*)]
\item For $\D=\Set$ or $\Pos$ we have $\Psi X_0 = X_0$ (discretely ordered). 
\item For $\D=\PSet$ the free pointed set on $X_0$ is $\Psi X_0 = X_0 + \{\bot\}$.
\item For $\D=\Inv$ the free involution algebra on $X_0$ is $\Psi X_0 = X_0 + \tl{X_0}$ where $\tl {X_0}$ is a copy of $X_0$ (whose elements are denoted $\tl x$ for $x\in X_0$). The involution swaps the copies of $X_0$, and the universal arrow $X_0\ra X_0+\tl{X_0}$ is the left coproduct injection.

\item For $\D=\SMod{\S}$ the free module $\Psi X_0$ is the submodule of $\S^{X_0}$ on all functions $X_0\ra \S$ with finite support. Equivalently, $\Psi X_0$ consists of formal linear combinations $\sum_{i=1}^n s_ix_i$ with $s_i\in \S$ and $x_i\in X_0$. In particular, $\Psi X_0 = \Pow_f X_0$ (finite subsets of $X_0$) for $\D = \JSL$, and $\Psi X_0$ is the vector space with basis $X_0$ for $\D=\Vect{\K}$.
\end{enumerate}
\end{example}

\begin{definition}\label{def:tensorproduct}
Given objects $A$, $B$ and $C$ of a variety $\D$, a \emph{bimorphism} from 
$A$, $B$ to 
$C$ is a function $f: \under{A}\times\under{B}\ra\under{C}$ such that the maps 
$f(a,\mathord{-}): \under{B}\ra \under{C}$ and $f(\mathord{-},b): \under{A}\ra 
\under{C}$ carry morphisms of $\D$ for every $a\in\under{A}$ and $b\in 
\under{B}$. A \emph{tensor product} of $A$ and $B$ is a universal bimorphism 
$\eta_{A,B}: \under{A}\times \under{B} \ra \under{A\t B}$, which means that 
for every bimorphism $f: \under{A}\times \under{B} \ra \under{C}$ there is a 
unique morphism $f': A\t B \ra C$ in $\D$ with $f'\o \eta_{A,B} = f$.
\end{definition}

\begin{theorem}[Banaschewski and Nelson \cite{bn76}]
Every commutative variety $\D$ of algebras or ordered algebras has tensor 
products, making
$(\D,\t,I)$ with $I=\Psi 1$ a symmetric monoidal closed category. That is,
we have the following bijective correspondence of morphisms, natural in 
$A,B,C\in \D$:
\[\begin{array}{rl}%
   f: & A\t B \ra C\\
    \hline
    \lambda f:& A\ra [B,C]
\end{array}
\]
\end{theorem}

\begin{remark}\label{rem:tensorprod}
We will use the following properties of the tensor product, see \cite{bn76}:
\begin{enumerate}
\item $- \t B$ and $A \t -$ are left-adjoints and hence 
preserve all colimits.
\item The morphism $\lambda f$ is obtained by currying the bimorphism $f' = 
f\o \eta_{A,B}: \under{A}\times\under{B}\to \under{C}$ corresponding to $f$, that is, $\lambda 
f(a)(b) = f'(a,b)$ for $a\in\under{A}$ and $b\in\under{B}$.
\item The evaluation morphism $\ev = \ev_{A,B}: [A,B]\t A \to B$,
i.e. the counit of the adjunction $- \t A \dashv [A,-]$, is the morphism in 
$\D$ corresponding to the bimorphism 
\[ \ev': \under{[A,B]}\times \under{A} \to \under{B},\quad (h,a)\mapsto 
h(a).\]
\item  The right unit isomorphism $i: A\cong A\t I$ is given by 
\[ i~=~ (\under{A} \xra{\langle{A,c_1}\rangle} 
\under{A}\times\under{I}\xra{\eta} \under{A\t 
I}) \]
where $c_1: \under{A}\to \under{I}$ is the constant map that sends $a\in A$ to 
the generator 
of $I=\Psi 1$.
\end{enumerate} 
\end{remark}

\begin{assumptions}\label{ass:global}
For the rest of this paper let $\D$ be a commutative variety of algebras or 
ordered algebras, equipped with the monoidal structure $(\D,\otimes,I)$ of the above theorem. Furthermore, we fix an 
object $X$ (of inputs) and an object $Y$ (of outputs) in $\D$.
\end{assumptions}

\begin{remark}
Despite the above algebraic assumptions we will state many of our 
definitions, theorems and proofs  for an arbitrary symmetric monoidal closed 
category
$(\D,\otimes,I)$ subject to 
additional properties. 
\end{remark}

\begin{remark}
  \label{rem:bullet}
Recall that a \emph{monoid} $(M,m,i)$ in a monoidal category $(\D,\otimes,I)$ 
(with tensor product $\otimes: \D\times\D\ra\D$ and tensor unit $I\in \D$) is 
an object $M$ equipped with a multiplication $m: M\t M\ra M$ and unit $i: I\ra 
M$ satisfying the usual associative and unit laws. Due to $\t$ and $I = \Psi 
1$ representing bimorphisms, this categorical definition is equivalent to the 
following algebraic one in our setting: a \emph{$\D$-monoid} is a triple 
$(M,\bullet, i)$ where $M$ is an object of $\D$ and $(\under{M},\bullet, i)$ 
is a monoid (in the usual sense) with $\bullet: \under{M}\times\under{M}\ra\under{M}$ a 
bimorphism of $\D$. A \emph{morphism} $h: (M,\bullet, i)\ra (M',\bullet', i')$ 
of $\D$-monoids is a morphism $h: M\ra M'$ of $\D$ such that $\under{h}: 
\under{M}\ra\under{M'}$ is a monoid morphism. We denote by $\DMon$ 
the category of $\D$-monoids and their homomorphisms. In the following we will 
freely work with $\D$-monoids in both categorical and algebraic disguise.
\end{remark}

\begin{example}\label{ex:dmon}
\begin{enumerate}[label=(\arabic*)]
\item In $\Set$ the tensor product is the cartesian product, $I=\{\ast\}$, and $\Set$-monoids are ordinary monoids.
\item In $\Pos$ the tensor product is the cartesian product, $I=\{\ast\}$, and $\Pos$-monoids are \emph{ordered monoids}, that is, posets equipped with a monoid structure whose multiplication $\bullet$ is monotone.
\item In $\PSet$ we have $I=\{\bot,\ast\}$, and the tensor product of pointed sets $(A,\bot_A)$ and $(B,\bot_B)$ is $A\t B = (A\setminus\{\bot_A\})\times (B\setminus\{\bot_B\}) + \{\bot\}$. $\PSet$-monoids are precisely monoids with zero. Indeed, given a $\PSet$-monoid structure on $(A,\bot)$ we have $x\bullet \bot= \bot = \bot\bullet x$ for all $x$ because $\bullet$ is a bimorphism, i.e.~$\bot$ is a zero element. Morphisms of $\Mon{\PSet}$ are zero-preserving monoid morphisms.
\item  An $\Inv$-monoid (also called an \emph{involution monoid}) is a monoid equipped with an  involution $x\mapsto \tl x$ such that $x\bullet \tl{y}= \tl{x}\bullet y = \tl{x\bullet y}$. For example, for any set $A$ the power set $\Pow A$ naturally carries the structure of an involution monoid: the involution takes complements, $\tl S = A\setminus S$, and the monoid multiplication is the symmetric difference $S\oplus T = (S\setminus T)\cup (T\setminus S)$. 
\item $\JSL$-monoids are precisely idempotent semirings (with $0$ and $1$). Indeed, a $\JSL$-monoid on a semilattice (i.e.~a commutative idempotent monoid) $(D,+,0)$ is given by a unit $1$ and a monoid multiplication that, being a bimorphism, distributes over $+$ and $0$.
\item More generally, a $\SMod{\S}$-monoid is precisely an associative algebra over $\S$: it consists of an $\S$-module together with a unit $1$ and a monoid multiplication that distributes over $+$ and $0$ and moreover preserves scalar multiplication in both components.
\end{enumerate}
\end{example}

\begin{notation}
We denote by $X^\cnum{n}$ ($n<\omega$) the $n$-fold tensor power of $X$, recursively defined by
$X^\cnum{0} = I$ and $X^\ecnum{n+1} = X \t X^\cnum{n}$.
\end{notation}

\begin{proposition}[see Mac Lane \cite{maclane}]\label{prop:freemon}
Let $\D$ be a symmetric monoidal closed category with countable coproducts. Then the forgetful functor $\DMon\ra\D$ has a left adjoint assigning to every object $X$ the free $\D$-monoid $\xs = \coprod_{n<\omega} X^\cnum{n}$. The monoid structure $(\xs,m_X,i_X)$ is given by the  coproduct injection $i_X: I=X^{\t 0} \to \xs$ and $m_X: \xs\t \xs\ra \xs$, where $\xs\t \xs \cong \coprod_{n,k<\omega} X^{\cnum{n}} \t X^{\cnum{k}}$ and $m_X$ has as its $(n,k)$-component the $(n+k)$-th coproduct injection. The universal arrow $\eta_X: X\ra \xs$ is the first coproduct injection.
\end{proposition}
\begin{remark}
  For any $\D$-monoid $(M,m,i)$ and any morphism $f: X \to M$ the extension to the unique $\D$-monoid morphism $f^+: \xs \to M$ with $f^+ \o \eta_X = f$ is defined inductively on the components of the coproduct $\coprod_{n<\omega} X^\cnum{n}$ as follows: $f^+ = [f^+_n]_{n < \omega}$, where $f_0^+ = i: X^\cnum{0} = I \to M$ and
\[
f^+_{n+1} = (\xymatrix@1{X^\cnum{(n+1)} = X \t X^\cnum{n} \ar[rr]^-{f \t f^+_n} && M \t M \ar[r]^-m & M}).
\]
\end{remark}
In our setting where $\D$ is commutative variety we have the following construction of free $\D$-monoids on free $\D$-objects:

\begin{proposition}\label{prop:freemon2}
The free $\D$-monoid on $X=\Psi X_0$ is $\xs = \Psi X_0^*$. Its monoid multiplication extends the concatenation of words in $X_0^*$, and its unit is the empty word $\epsilon$.
\end{proposition}

\begin{proof}
A constructive proof can be found in \cite{ammu14}. Here we sketch a more conceptual argument, using the universal property of the tensor product. Observe that the free-algebra functor $\Psi: \Set\ra\D$ is strongly monoidal, i.e.  it preserves the unit and tensor product up to natural isomorphism. Indeed, we have $\Psi 1 = I$ by definition. To see that $\Psi(A\times B) \cong \Psi A \t \Psi B$ for all sets $A$ and $B$, consider the following bijections (natural in $D\in\D$):
\begin{align*}
 \D(\Psi(A\times B), D) &\cong \Set(A\times B, \under{D})\\
 &\cong \Set(A,\under{D}^B) \\
 &\cong \Set(A, |[\Psi B, D]|)\\
 &\cong \D(\Psi A, [\Psi B, D]) \\
 &\cong \D(\Psi A\t \Psi B, D).
\end{align*}
This implies $\Psi(A\times B) \cong \Psi A \t \Psi B$ by the Yoneda lemma. Using the fact that $\Psi$ preserves coproducts, being a left adjoint, we conclude
\[ \xs \cong \coprod_{n<\omega} X^{\cnum{n}} \cong \coprod_{n<\omega} \Psi X_0^n \cong \Psi(\coprod_{n<\omega} X_0^n) = \Psi X_0^*.\tag*{\qEd}\]
\def\popQED{}
\end{proof}

\begin{example}
\begin{enumerate}[label=(\arabic*)]
\item In $\Set$ and $\Pos$ we have $\xs = X^*$ (discretely ordered). 
\item In $\PSet$ with $X=\Psi X_0 = X_0 + \{\bot\}$ we get $\xs = X_0^* +\{\bot\}$. The product $x\bullet y$ is concatenation for $x,y\in X_0^*$, and otherwise $\bot$.

\item In $\Inv$ with $X = \Psi X_0 = X_0 + \tl{X_0}$ we have $\xs = X_0^* + \tl{X_0^*}$. The multiplication restricted to $X_0^*$ is concatenation, and is otherwise determined by $\tl u\bullet v = \tl{uv} = u \bullet \tl v$ for $u,v\in X_0^*$.

\item In $\JSL$ with $X=\Psi X_0 = \Pow_f X_0$ we have $\xs = \Pow_f X_0^*$, the semiring of all finite languages over $X_0$. Its addition is union and its multiplication is the concatentation of languages.

\item More generally, in $\SMod{\S}$ with $X=\Psi X_0$ we get $\xs = \Psi X_0^* = \S[X_0]$, the module of all polynomials over $\S$ in finitely many
variables from $X_0$, i.e., $\sum_{i=1}^n c(w_i)w_i$ with  $w_i\in X_0^*$ and $c(w_i)\in \S$. Hence the elements of $\S[X_0]$ are functions $c: X_0^*\ra \S$ with finite support, i.e.~finite weighted languages. The $\S$-algebraic structure of $\S[X_0]$ is given by the usual addition, scalar multiplication and product of polynomials.
\end{enumerate}
\end{example}

\begin{definition}[Goguen \cite{goguen75}]
\label{def:aut}
A \emph{$\D$-automaton} $(Q,\delta,i,f)$ consists of an object $Q$ (of states) and morphisms $\delta: X\t Q \ra Q$, $i: I\ra Q$ and $f: Q\ra Y$; see Diagram~(\ref{eq:daut}). An \emph{automata homomorphism} $h: (Q,\delta,i,f) \ra (Q',\delta',i',f')$ is a morphism $h: Q\ra Q'$ preserving transitions as well as initial states and outputs, i.e.~making the following diagrams commute:
\[  
\xymatrix{
X\t Q \ar[d]_{X\t h}\ar[rr]^\delta && Q \ar[d]^h \\
X\t Q' \ar[rr]_{\delta'}  && Q' 
}
\quad\quad\quad
\xymatrix{
I \ar[drr]_{i'} \ar[rr]^i && Q \ar[d]^h \ar[rr]^f  && Y \\
 && Q'\ar[urr]_{f'}  &&
}
\]
\end{definition}

The above definition makes sense in any monoidal category $\D$. In our algebraic setting, since $I=\Psi 1$, the morphism $i$ chooses an \emph{initial state} in $\under{Q}$. Moreover, if  $X=\Psi X_0$ for some set $X_0$ (of inputs), the morphism $\delta$ amounts to a choice of endomorphisms $\delta_a: Q\ra Q$ for $a\in X_0$, representing transitions. This follows from the bijections
\[\begin{array}{ll}%
    \Psi X_0 \t Q \ra Q & \text{in $\D$}\\
    \hline
    \Psi X_0 \ra [Q,Q]& \text{in $\D$}\\
    \hline
     X_0 \ra \D(Q,Q) & \text{in $\Set$}%
\end{array}
\]
We will occasionally write $q\xra{a} q'$ for $\delta_a(q)=q'$ if $\delta$ is clear from the context.
\begin{example}\label{ex:automata}
\begin{enumerate}[label=(\arabic*)]
\item The classical deterministic automata are the case $\D=\Set$ and $Y=\{0,1\}$. Here $f: Q\ra \{0,1\}$ defines the set $F=f^{-1}[1]\seq Q$ of final states. For a general set $Y$ we get deterministic Moore automata with outputs in $Y$.

\item For $\D=\Pos$ and $Y=\{0<1\}$ we get \emph{ordered} deterministic automata. That is, automata with a partially ordered set $Q$ of states, monotone transition maps $\delta_a$, and an upper set $F=f^{-1}[1]\seq Q$ of final states.

\item The setting $\D=\PSet$ with $X = X_0 + \{\bot\}$ and $Y=\{\bot,1\}$ gives \emph{partial} deterministic automata. Indeed, the state object $(Q,\bot)$ has transitions $\delta_a: (Q,\bot)\ra (Q,\bot)$ for $a\in X_0$ preserving $\bot$, that is, $\bot$ is a sink state. Equivalently, we may consider $\delta_a$ as a partial transition map on the state set $Q\setminus\{\bot\}$. The morphism $f: (Q,\bot)\ra \{\bot,1\}$ again determines a set of final states $F=f^{-1}[1]$ (in particular, $\bot$ is non-final). And the morphism $i:\{\bot,\ast\}\ra (Q,\bot)$ determines a partial initial state: either $i(\ast)$ lies in $Q\setminus \{\bot\}$, or no initial state is defined.

\item In $\D=\Inv$ let us choose $X=X_0+\tl{X_0}$ and $Y=\{0,1\}$ with $\tl 0 = 1$. An $\Inv$-automaton is a deterministic automaton with complementary states $x\mapsto \tl x$ such that (i) for every transition $p\xra{a} q$ there is a complementary transition $\tl p \xra{a} \tl{q}$ and (ii) a state $q$ is final iff $\tl q$ is non-final.

\item For $\D=\JSL$ with $X=\Pow_f X_0$ and $Y=\{0,1\}$ (the two-chain) an automaton consists of a semilattice $Q$ of states, transitions $\delta_a: Q\ra Q$ for $a\in X_0$ preserving finite joins (including $0$), an initial state $i\in Q$ and a homomorphism $f: Q\ra\{0,1\}$ which defines a \emph{prime upset} $F=f^{-1}[1]\seq Q$ of final states. Primality means that a finite join of states is final iff one of the states is. In particular, $0$ is non-final. For example, the classical determinization of a nondeterministic automaton in $\Set$ via the powerset construction yields a $\JSL$-automaton, where the semilattice structure is given by finite union.  

\item More generally, automata in $\D=\SMod{\S}$ with $X=\Psi X_0$ and $Y=\S$ are \emph{$\S$-weighted automata}. Such an automaton consists of an $\S$-module $Q$ of states, linear transitions $\delta_a: Q\ra Q$ for $a\in X_0$, an initial state $i\in Q$ and a linear output map $f: Q\ra \S$. 

\end{enumerate}
\end{example}

\begin{remark}\label{rem:algcoalg}
\begin{enumerate}[label=(\arabic*)]
\item An \emph{algebra} for an endofunctor $F$ of $\D$ is a pair $(Q,\alpha)$ consisting of an object $Q$ and a morphism $\alpha: FQ\ra Q$. A \emph{homomorphism} $h: (Q,\alpha)\ra(Q',\alpha')$ of $F$-algebras is a morphism $h: Q\ra Q'$ with $h\o \alpha = \alpha'\o Fh$. Throughout this paper we work with the endofunctor 
\[FQ = I + X\t Q.\]
 Its algebras are denoted as triples $(Q,\delta,i)$ with $\delta: X\t Q\ra Q$ and $i: I\ra Q$. Hence $\D$-automata are precisely 
$F$-algebras equipped with an output morphism $f: Q\ra Y$. Moreover, automata homomorphisms are precisely $F$-algebra homomorphisms preserving outputs.

\item Analogously, a \emph{coalgebra} for an endofunctor $T$ of $\D$ is a pair $(Q,\gamma)$ consisting of an object $Q$ and a morphism $\gamma: Q\ra TQ$. A \emph{homomorphism} $h: (Q,\gamma)\to(Q',\gamma')$ of $T$-coalgebras is a morphism $h: Q\to Q'$ with $Th\o \gamma = \gamma'\o h$. Throughout this paper we work with the endofunctor
\[TQ = Y\times [X,Q].\] Its coalgebras are denoted as triples $(Q,\tau,f)$ with $\tau: Q\ra [X,Q]$ and $f:Q\ra Y$. Hence $\D$-automata are precisely \emph{pointed} $T$-coalgebras, i.e.~$T$-coalgebras equipped with a morphism $i: I\ra Q$. Indeed, given a pointed coalgebra $I \xra{i} Q \xra{\langle f,\tau\rangle} Y\times[X,Q]$, 
the morphism $Q\xra{\tau} [X,Q]$ is the curried form of a morphism $Q\t X \xra{\cong} X\t Q \xra{\delta} Q$. Hence automata homomorphisms are precisely the $T$-coalgebra homomorphisms preserving initial states.
\end{enumerate}
\end{remark}

\begin{definition}
  \label{def:asso}
Given a $\D$-monoid $(M,m,i)$ and a morphism $e: X\ra M$ of $\D$, the \emph{$F$-algebra associated to $M$ and $e$} has carrier $M$ and structure 
\[
  \delta~\equiv~(X\t M \xra{e\t M} M\t M \xra{m} M) \quad\text{and}\quad i: I\to M.
\] 
\end{definition}
In particular, the $F$-algebra associated to the free monoid $\xs$ (and its universal arrow $\eta_X$) is
\[
  \delta_X~\equiv~(X\t \xs \xra{\eta_X\t \xs} \xs\t \xs \xra{m_X} \xs) \quad\text{and}\quad i_X: I\to \xs.
\]  

\begin{example} 
In $\Set$ every monoid $M$ together with an ``input'' map $e: X\ra M$ determines an $F$-algebra with initial state $i$ and transitions $\delta_a =  e(a)\bullet \mathord{-}$ for all $a\in X$. The $F$-algebra associated to $X^*$ is the automaton of words: its initial state is $\epsilon$ and the transitions are given by $w\xra{a} aw$ for $a\in X$.
\end{example}

\begin{proposition}[Goguen \cite{goguen75}]\label{prop:initalg}
For any symmetric monoidal closed category $\D$ with countable coproducts, $\xs$ is the initial algebra for $F$.
\end{proposition}
\takeout{
In more detail, for $\xs = \coprod_{n<\omega} X^{\cnum{n}}$ we have, since $\mathord{-}\t X$ preserves coproducts (being a left adjoint), that 
$X \t \xs \cong \coprod_{n\geq 1} X^{\cnum{n}}$. Then $\xs$ is an $F$-algebra due to the following isomorphism (where $I=X^{\cnum{0}}$):
\[ F\xs \cong I + X\t \xs \cong \coprod_{n<\omega} X^{\cnum{n}} = \xs. \]
For every $F$-algebra $Q$ the unique homomorphism $e_Q: \xs = \coprod_{n<\omega} X^{\cnum{n}} \ra Q$ has components $e_Q^{\cnum{n}}\ra Q$ defines by the following recursion:
\[ e_Q^{\cnum{0}}= i: I\ra Q\quad\text{and}\quad e_Q^{\ecnum{n+1}} = (X \t X^\cnum{n} \xra{X \t e_Q^\cnum{n}} X\t Q \xra{\delta} Q). \]
}
\begin{remark}\label{rem:concept}
\begin{enumerate}  \item For any $F$-algebra $(Q,\delta,i)$ the unique 
$F$-algebra homomorphism $e_Q: \xs \ra Q$ is constructed as follows: extend 
the morphism $\lambda \delta: X \to [Q,Q]$ to a $\D$-monoid morphism 
$(\lambda\delta)^+: \xs \to [Q,Q]$. Then
  \begin{equation}\label{eq:eA}
    e_Q = (\xs \cong \xs \t I \xra{(\lambda \delta)^+ \t i} [Q,Q] \t Q \xra{\ev} Q),
  \end{equation}
where $\ev$ is the `evaluation morphism', i.e.~the counit of the adjunction $- \t Q \dashv [Q,-]$. 

To see this, recall that in any symmetric monoidal closed category $\D$ the object $[Q,Q]$ can be equipped with a $\D$-monoid structure: the unit $j: I \to [Q,Q]$ and multiplication $m: [Q,Q] \t [Q,Q] \to [Q,Q]$ are the unique morphisms making the following diagrams commutative (where $\iota'_Q: Q \to I \t Q$ is the left unit isomorphism):
\[
\xymatrix{
  [Q,Q] \t Q \ar[r]^-{\ev} & Q \\
  I \t Q 
  \ar[u]^{j \t Q}
  \ar[ru]_{\iota'_Q}
}
\qquad
\xymatrix@C+1pc{
  [Q,Q] \t Q \ar[r]^-{\ev} & Q \\
  [Q,Q] \t [Q,Q] \t Q
  \ar[u]^{m \t Q}
  \ar[r]_-{[Q,Q]\t \ev}
  &
  [Q,Q] \t Q \ar[u]_{\ev}
}
\]
Then the following commutative diagram shows that $e_Q$ is an $F$-algebra homomorphism, as claimed:
\[
\xymatrix@C+2pc{
  X \t \xs 
  \ar[r]^-{\eta_X \t \xs} 
  \ar[d]^{X \t \iota_{\xs}}
  &
  \xs \t \xs 
  \ar[r]^-{m_X}
  \ar[d]^{\xs \t \iota_{\xs}}
  &
  \xs
  \ar[d]^{\iota_{\xs}}
  &
  I
  \ar[l]_-{i_X}
  \ar[d]_{\iota_I}
  \ar `r[d] `[ddd] [lddd]^(.4){i}
  \\
  X \t \xs \t I
  \ar[r]^-{\eta_X \t \xs\t I}
  \ar[d]^{X \t (\lambda\delta)^+ \t i}
  &
  \xs \t \xs \t I
  \ar[r]^-{m_X \t I}
  \ar[d]^{(\lambda\delta)^+ \t (\lambda\delta)^+ \t i}
  &
  \xs \t I
  \ar[d]^{(\lambda\delta)^+ \t i}
  &
  I \t I
  \ar[l]_-{i_X \t I}
  \ar[d]_{I \t i}
  \\
  X \t [Q,Q] \t Q
  \ar[r]^-{\lambda\delta \t [Q,Q] \t Q}
  \ar[d]^{X \t \ev}
  &
  [Q,Q] \t [Q,Q] \t Q
  \ar[r]^-{m \t Q}
  \ar[d]_{[Q,Q]\t \ev}
  &
  [Q,Q] \t Q
  \ar[d]^{\ev}
  &
  I \t Q
  \ar[l]_-{j \t Q}
  \ar[ld]^{(\iota'_Q)^{-1}}
  \\
  X\t Q
  \ar[r]^-{\lambda\delta \t Q}
  &
  [Q,Q]\t Q
  \ar[r]^-{\ev}
  &
  Q
  \ar@{<-} `d[l] `[ll]^\delta [ll]
  &
}
\]
\item If $\D$ is a commutative variety of algebras or ordered algebras, the 
monoid 
structure on $[Q,Q]$ is given by composition and the identity map. Let 
$\delta_x: 
Q\to Q$ denote the image of $x\in\xs$ under $(\lambda\delta)^+: \xs\to [Q,Q]$. 
Then the initial $F$-algebra homomorphism $e_Q$ in \eqref{eq:eA} sends $x$ to 
the state $\delta_x\o i: I\to Q$. This follows from the commutative 
diagram below, see Remark \ref{rem:tensorprod}:
\[
\xymatrix{
\under{\xs} \ar[r]^\cong \ar[dr]_{\langle \xs,c_1\rangle} & \under{\xs\t I} 
\ar[r]^{(\lambda\delta)^+\t i} & \under{[Q,Q]\t Q} \ar[r]^<<<<<\ev & 
\under{Q} \ar@{<-} `u[l] `[lll]_{e_Q} [lll]\\
& \under{\xs}\times \under{I} \ar[r]_{(\lambda\delta)^+\times i} \ar[u]_\eta 
& \under{[Q,Q]}\times \under{Q} \ar[ur]_{\ev'} \ar[u]^\eta &
}
\]
\end{enumerate}
\end{remark}

\begin{notation}\label{not:deltastar}
 $\delta^\cst: \xs\t Q\ra Q$ denotes the uncurried form of $(\lambda \delta)^+: \xs\ra[Q,Q]$.
\end{notation}
\begin{remark}
\label{rem:aux}
\begin{enumerate}[label=(\arabic*)]
  \item One can also define $\delta^\cst$ explicitly (using that $-\t Q$ preserves the coproduct $\xs = \coprod_{n < \omega} X^\cnum{n}$) as $\delta^\cst = [\delta^\cst_n]_{n < \omega}$ where $\delta^\cst_0 = (\xymatrix@1{I \t Q \ar[r]^-\cong & Q})$ and
\[
\delta^\cst_{n+1} = (\xymatrix@1{%
X^\cnum{(n+1)} \t Q = X \t X^\cnum{n} \t Q \ar[rr]^-{X \t \delta^\cst_n} && X \t Q \ar[r]^-\delta & Q%
}).
\]
From this definition we clearly have the equation
\[
\delta = (\xymatrix@1{X \t Q \ar[rr]^-{\eta_X \t Q} && \xs \t Q \ar[r]^-{\delta^\cst} & Q}).
\]
  \item By \cite[Theorem 3.3]{goguen75} the morphism $\delta^\cst: \xs \t Q \to Q$ defines an \emph{action} of the monoid $\xs$ on $Q$, i.e.~the following diagrams commute:
    \[
      \xymatrix@C+1pc{
        I \t Q \ar[r]^{i_X \t Q} \ar[rd]_{\cong} & \xs \t Q \ar[d]^{\delta^\cst} \\
        & Q
      }
      \quad\text{and}\quad
      \xymatrix@C+1pc{
        \xs \t \xs \t Q \ar[r]^-{m_X \t Q} \ar[d]_{\xs \t \delta^\cst} & \xs \t Q \ar[d]^{\delta^\cst} \\
        \xs \t Q \ar[r]_-{\delta^\cst} & Q
      }
    \]
 Moreover,  every $F$-algebra homomorphism $h$ from $(Q, \delta, i)$ to $(Q', \delta', i')$ is a morphism of the corresponding monoid actions, i.e.~the following diagram commutes:
     \[
       \xymatrix{
         \xs \t Q \ar[r]^-{\delta^\cst} \ar[d]_{\xs \t h} & Q \ar[d]^h & I \ar[l]_-i \ar[ld]^{i'}\\
         \xs \t Q' \ar[r]_-{(\delta')^\cst} & Q'
     }
     \]
  \item Furthermore, it is not difficult to prove that for $\delta_X : X \t \xs \to \xs$ coming from the $F$-algebra associated to $\xs$ we have that $\delta_X^\cst$ is the monoid multiplication: 
  \[ \delta_X^\cst = m_X : \xs \t \xs \to \xs.\]
\end{enumerate}
\end{remark}
\begin{remark}
 Recall from Rutten \cite{rutten} that the final coalgebra for the functor $TQ = \{0,1\}\times Q^X$ on $\Set$ is carried by the set $\Pow X^*\cong [X^*,\{0,1\}]$ of all  languages over $X$. The transitions are given by taking \emph{left derivatives}, 
 \[ L \xra{a} a^{-1}L = \{w\in X^*: aw\in L\}\quad\text{for $L\in\Pow X^*$ and $a\in X$},\]
 and the output map $\Pow X^* \to \{0,1\}$ sends $L$ to $1$ iff $L$ contains the empty word. Given any coalgebra $Q$, the unique coalgebra homomorphism from $Q$ to $\Pow\Sigma^*$ assigns to every state $q$ the language accepted by $q$ (as an initial state). 
 
  These observations generalize to arbitrary symmetric monoidal closed categories. The object $[\xs, Y]$ of $\D$  carries the following $T$-coalgebra structure: its transition morphism $\tau_{\fc}:\fc\ra [X,\fc]$ is the two-fold curryfication of 
\[ \fc\t X \t \xs \xra{\fc\t \eta_X \t \xs} \fc\t \xs \t \xs \xra{\fc\t m_X} \fc\t \xs \xra{\ev} Y, \]
and its output morphism $f_{\fc}: \fc\ra Y$ is 
\[ f_{[\xs, Y]} = (\fc\cong \fc\t I \xra{\fc\t i_X} \fc\t \xs  \xra{\ev} Y).\] 
\end{remark}

\begin{proposition}\label{prop:fincoalg}
For any symmetric monoidal closed category $\D$, the coalgebra $\fc$ is the final coalgebra for $T$.
\end{proposition}

\begin{proof}
Given a $T$-coalgebra $(Q,\tau,f)$, consider the morphism $\delta = (X\t Q \xra{\cong} Q\t X \xra{\beta} Q)$ where $\beta$ is the uncurried form of $\tau: Q\ra[X,Q]$, and denote by $\delta^\cst: \xs\t Q\ra Q$ its extension, see Notation \ref{not:deltastar}. We claim that the unique coalgebra homomorphism into $\fc$ is $\lambda h: Q\ra \fc$, where 
\begin{equation}\label{eq:finalmorphism}
 h = (Q\t \xs \cong \xs \t Q \xra{\delta^\cst} Q \xra{f} Y).
 \end{equation}
Let us first prove that $h$ is indeed a coalgebra homomorphism. Preservation of outputs is shown by the following commutative diagram (for the upper left-hand part use Remark~\ref{rem:aux}(2)):
\[
\xymatrix@C+2pc{
Q \ar[ddr]^\cong \ar[rd]^\cong \ar@{=}[rr] \ar[ddd]_{\lambda h} & & Q \ar[r]^{f}  & Y\\
& I \t Q \ar[r]^-{i_X\t Q}& \xs\t Q   \ar[u]_{\delta^\cst}&\\
& Q\t I \ar[r]_{Q\t i_X} \ar[d]^{\lambda h\t I} \ar[u]^\cong & Q\t \xs \ar[dr]^{\lambda h\t \xs} \ar[uur]_h \ar[u]_{\cong}  &\\
\fc \ar[r]_{\cong}& \fc\t I \ar[rr]_{\fc\t i_X} && \fc\t \xs \ar[uuu]_\ev 
}
\]
For preservation of transitions it suffices to show that the following diagram commutes, where $\ol\tau: \fc\t  X\ra \fc$ is the uncurried coalgebra structure of $\fc$:
\[  
\xymatrix{
Q \t X \ar[d]_{\lambda h \t X} \ar[r]^{\beta} & Q \ar[d]^{\lambda h}\\
\fc \t X \ar[r]_{\ol\tau} & \fc
}
\]
But the above diagram is precisely the curried version of the following one, which commutes by the properties listed in Remark \ref{rem:aux}. (We omit writing $\otimes$ for space reasons.)
\[
\xymatrix{
Q X 
\xs \ar[rrrrr]^{\beta\xs} \ar[ddr]_{Q\eta_X\xs} \ar[ddd]_{\lambda h X\xs} \ar[rd]^\cong
& & & && Q \xs \ar[dl]_{\cong} \ar[ddd]^h\\
&
XQ\xs \ar[rrrru]^-{\delta \xs} \ar[r]_-{\eta_X Q \xs} 
&
\xs Q \xs \ar[ld]^(.4)\cong \ar[r]_\cong \ar[rrru]_(.6){\delta^\cst \xs}
&
\xs\xs Q 
\ar[r]_-{\xs \delta^\cst}
& \xs Q \ar[d]^{\delta^\cst} &\\
& Q \xs \xs \ar[d]^{\lambda h\xs\xs} \ar[r]^{Q m_X} & Q \xs \ar[drrr]^h \ar[d]^{\lambda h \xs} \ar[r]^\cong & \xs Q \ar[r]^{\delta^\cst} & Q \ar[dr]^{f} &\\
\fc X\xs \ar[r]^-{\begin{turn}{90}$\labelstyle\fc\eta_X\xs$\end{turn}} & \fc \xs\xs \ar[r]^(.6){\begin{turn}{45}$\labelstyle\fc m_X$\end{turn}} & \fc\xs \ar[rrr]_{\ev} & && Y
}
\]
For the uniqueness, suppose that any coalgebra homomorphism $\lambda h:Q\ra \fc$ is given. We show that $h: Q\t \xs\ra Y$ is determined by the composites $(Q \t X^{\t n} \xra{Q\t i_n} Q\t X^\cst \xra{h} Y)$, $n<\omega$,
where $i_n: X^{\t n}\ra \xs$ is the $n$-th coproduct injection. This proves the uniqueness of $\lambda h$: since $\t$ preserves coproducts, the morphisms $(Q \t i_n)_{n<\omega}$ form a coproduct cocone. For $n=0$, the claim is proved by the following diagram:
\[
\xymatrix{
Q\t I \ar[ddd]_{Q\t i_0} \ar[rrr]^{\cong} \ar[dr]^{\lambda h \t I} & & & Q \ar[ddd]^{f} \ar[dl]_{\lambda h}\\
& \fc \t I \ar@<2pt>[r]^-\cong \ar[d]_{\fc \t i_0} & \fc \ar@<2pt>[l]^-\cong \ar[rdd]_{f_{[\xs, Y]}}&\\
& \fc\t \xs \ar[drr]_{\ev} & &\\
Q\t \xs \ar[rrr]_h \ar[ur]^{\lambda h \t \xs} &&& Y 
}
\]
 And the following diagram shows that $h\o (Q\t i_{n+1})$ is determined by $h\o (Q\t i_n)$ (again we omit $\t$, in particular we write $X^n$ for $X^\cnum{n}$):
 \[
 \xymatrix{
 QXX^n \ar[dddd]_{\beta X^n} \ar[rrr]^{Qi_{n+1}} \ar[dr]^(.6)*+{\labelstyle \lambda hXX^n} & & & Q\xs \ar[rr]^h \ar[d]^{\lambda h \xs} & & Y\\
   & \fc XX^n \ar[dd]_{\ol\tau X^n}  \ar[rr]^{\fc i_{n+1}} \ar[dr]^(.6)*+{\labelstyle \fc Xi_n} & & \fc\xs \ar[urr]^\ev & & \\
   & & \fc X \xs \ar[r]^(.65){\begin{turn}{25}$\labelstyle\fc \eta_X \xs$\end{turn}} \ar[dr]_{\ol\tau\xs} & \fc\xs\xs \ar[u]_{\fc m_X} \ar@{}[d]|{(*)}& &\\
   & \fc X^n \ar[rr]_{\fc i_n} & & \fc\xs \ar[uuurr]_\ev & &\\
 QX^n \ar[ur]_{\lambda h X^n} \ar[rrrrr]_{Qi_n} & & & & & Q\xs \ar[ull]_{\lambda h \xs} \ar[uuuu]_h  
 }
 \]
Note that part marked with $(*)$ commutes by the definition of the coalgebra structure on $[X^\cst, Y]$ and all other parts are easy to see.
\end{proof}

\begin{remark}\label{rem:finalmorphism}
If $\D$ is a commutative variety of algebras or ordered algebras, the final 
homomorphism $\lambda h: Q\to [\xs,Y]$ sends a state $q\in \under{Q}$ to the 
morphism $x\mapsto f\o \delta_x(q)$, with $\delta_x: Q\to Q$  defined as in 
Remark \ref{rem:concept}(2). To see this, consider the commutative diagram 
below, 
where $(\delta^\cst)'$ is the bimorphism corresponding to $\delta^\cst$:
\[
\xymatrix{
\under{Q\t \xs}  \ar[r]^\cong & \under{\xs \t Q} \ar[r]^{\delta^\cst} & 
\under{Q} \ar[r]^f & \under{Y} \ar@{<-} `u[l] `[lll]_h [lll]\\
\under{Q}\times\under{\xs} \ar[u]^\eta \ar[r]_\cong & \under{\xs}\times 
\under{Q} \ar[u]^\eta \ar[ur]_{(\delta^\cst)'} &&
}
\]
By the definition of $\delta_x$ and Remark \ref{rem:tensorprod}, the 
bimorphism 
$(\delta^\cst)'$ maps $(x,q)$ to $\delta_x(q)$. Thus $h\o \eta$ maps $(q,x)$ 
to $f\o\delta_x(q)$. This implies the claim since $\lambda h$ is the curried 
form of $h\o \eta$.
\end{remark}
Proposition \ref{prop:fincoalg} motivates the following definition:

\begin{definition}[Goguen \cite{goguen75}]
 A \emph{language} in $\D$ is a morphism $L:\xs\ra Y$.
\end{definition}

Note that if $X=\Psi X_0$ (and hence $\xs=\Psi X_0^*$) for some set $X_0$, one can identify a language $L: \xs  \to Y$ in $\D$ with its adjoint transpose $\tl L: X_0^*\ra \under{Y}$, using the adjunction $\Psi\dashv |\mathord{-}|: \D\ra \Set$. In the case where $\under{Y}$ is a two-element set, $\tl L$ is the characteristic function of a ``classical'' language $L_0 \seq X_0^*$.

\begin{example}\label{ex:language}
\begin{enumerate}[label=(\arabic*)]
\item In $\D = \Set$ (with $\xs = X^*$ and $Y = \{ 0, 1\}$) one represents $L_0\seq X^*$  by its characteristic function $L: X^*\ra\{0,1\}$. Analogously for $\D = \Pos$ with $Y=\{0<1\}$.

\item In $\D=\PSet$ (with $X=X_0 + \{\bot\}$, $\xs = X_0^* + \{\bot\}$ and $Y=\{\bot,1\}$) one represents $L_0\seq X_0^*$  by its extended characteristic function $L: X_0^*+\{\bot\}\ra \{\bot,1\}$ where $L(\bot)=\bot$.

\item In $\D=\Inv$ (with $X=X_0+\tl{X_0}$, $\xs=X_0^* + \tl{X_0^*}$ and $Y=\{0,1\}$) one represents $L_0\seq X_0^*$  by $L: X_0^* + \wt{X_0^*} \to \{0,1\}$ where $L(w) = 1$ iff $w \in L_0$ and $L(\wt w) = 1$ iff $w \not\in L_0$ for all words $w \in X_0^*$.

\item In $\D=\JSL$ (with $X=\Pow_f X_0$, $\xs = \Pow_f X_0^*$ and $Y=\{0,1\}$) one represents $L_0\seq X_0^*$ by $L: \Pow_f X_0^*\ra\{0,1\}$ where $L(U)=1$ iff $U\cap L_0 \neq \emptyset$.

\item In $\D=\SMod{\S}$ (with $X=\Psi X_0$, $\xs = \S[X_0]$ and $Y=\S$) an $\S$-weighted language $L_0: X_0^*\ra \S$ is represented by its free extension to a module homomorphism
\[
L:\S[X_0] \to \S,\quad L\left(\sum\limits_{i =1}^{n} c(w_i)w_i \right) = \sum_{i=1}^n c(w_i)L_0(w_i).
\]
\end{enumerate}
\end{example}

\begin{definition}[Goguen \cite{goguen75}]
The language \emph{accepted} by a $\D$-automaton $(Q,\delta,i,f)$ is $L_Q = (\xs \xra{e_Q} Q \xra{f} Y )$, where $e_Q$ is the $F$-algebra homomorphism of Remark \ref{rem:concept}.
\end{definition}

\begin{example}\label{ex:langacc}
\begin{enumerate}[label=(\arabic*)]
\item In $\D=\Set$ or $\Pos$ with $Y=\{0,1\}$, the homomorphism $e_Q: X^* \ra Q$ assigns to every word $w$ the state it computes in $Q$, i.e.~the state the automaton reaches on input $w$. Thus $L_Q(w)=1$ iff $Q$ terminates in a final state on input $w$, which is precisely the standard definition of the accepted language of an (ordered) automaton. For general $Y$, the function $L_Q: X^*\ra Y$ is the behavior of the (ordered) Moore automaton $Q$, i.e.~$L_Q(w)$ is the output of the last state in the computation of $w$.

\item For $\D=\PSet$ with $X=X_0+\{\bot\}$ and $Y=\{\bot,1\}$, we have $e_Q: X_0^* + \{\bot\}\ra (Q,\bot)$ sending $\bot$ to $\bot$, and sending a word in $X_0^*$ to the state it computes (if any), and to $\bot$ otherwise. Hence $L_Q: X_0^*+\{\bot\}\ra \{\bot,1\}$ defines (via the preimage of $1$) the usual language accepted by a partial automaton.

\item In $\D=\Inv$ with $X=X_0+\tl{X_0}$ and $Y= \{0,1\}$, the map 
$L_Q: X_0^*+ \tl{X_0^*} \ra \{0,1\}$ sends $w\in X_0^*$ to $1$ iff $w$ computes a final state, and it sends $\tl w\in \tl{X_0^*}$ to $1$ iff $w$ computes a non-final state. 

\item In $\D=\JSL$ with $X=\Pow_f X_0$ and $Y=\{0,1\}$, the map $L_Q: \Pow X_0^* \ra \{0,1\}$ assigns to $U\in \Pow_f X_0^*$ the value $1$ iff the computation of at least one word in $U$ ends in a final state.

\item In $\D=\SMod{\S}$ with $X=\Psi X_0$ and $Y=\S$, the map $L_Q: \S[X_0] \ra \S$ assigns to $\sum_{i=1}^n c(w_i)w_i$ the value $\sum_{i=1}^n c(w_i)y_i$, where $y_i$ is the output of the state the automaton computes on input $w_i$. Taking $Q = \S^n$ for some natural number $n$ yields a classical $n$-state weighted automaton; indeed, $i \in \S^n$ is an initial vector, the linear map $f: \S^n \to \S$ corresponds to an output vector $o \in \S^n$ and the linear transition map $\delta: X \t \S^n \to \S^n$ is given by a family of linear maps $(\delta_a: \S^n \to \S^n)_{a \in X_0}$, which can represented by a family $(M_a)_{a \in X_0}$ of $n \times n$ matrices over $\S$. It then follows that the restriction of $L_Q$ to a map $X_0^* \to \S$ is is the usual weighted language assigning to a word $w$ the element $o \cdot M_w \cdot i^T$ of $\S$, where $M_w$ is the obvious product of the matrices $M_a$.
\end{enumerate}
\end{example}

\begin{remark}
By Remark \ref{rem:algcoalg} every $\D$-automaton $(Q,\delta,i,f)$ defines an $F$-algebra as well as a $T$-coalgebra. Our above definition of $L_Q$ was purely algebraic. The corresponding coalgebraic definition uses the unique coalgebra homomorphism
$c_Q: Q\ra \fc$ into the final $T$-coalgebra and precomposes with $i: I\ra Q$ to get a morphism $c_Q \o i: I\ra \fc$ (choosing a language, i.e.~an element of $\fc$). Unsurprisingly, the results are equal:
\end{remark}

\begin{proposition}\label{prop:lang}
The language $L_Q: \xs \ra Y$ of an automaton $(Q,\delta,i,f)$ is the uncurried form of the morphism $c_Q \o i: I\ra \fc$.
\end{proposition}

\begin{proof}
Recall from the proof of Proposition \ref{prop:fincoalg} that the uncurried version of $c_Q$ is the morphism
\[ Q\t \xs \cong \xs \t Q \xra{\delta^\cst} Q \xra{f} Y.\]
Hence $c_Q \o i: I\ra \fc$ determines the language
\[ \xs\cong \xs\t I \xra{\xs\t i} \xs\t Q \xra{\delta^\cst} Q \xra{f} Y,\]
and this is precisely $L_Q$, as shown by the diagram below (where $\delta_X$ is given by the $F$-algebra structure associated to $\xs$ and $\eta_X$, see Remark~\ref{def:asso}):
\[
\xymatrix@C+2pc{
\xs   \ar[r]^-\cong 
& 
\xs\t I \ar[r]^{\xs \t i_X} \ar[dr]_-{\xs\t i} 
& 
\xs\t \xs \ar[d]^{\xs \t e_Q} \ar[r]^-{\delta_X^\cst} 
& 
\xs \ar[d]^{e_Q} \ar[r]^{L_Q}
\ar@{<-} `u[l] `[lll]_\id [lll]  
& 
Y
\\
&  & \xs \t Q \ar[r]_{\delta^\cst} & Q \ar[ur]_{f} &
}
\]
Indeed, the right-hand triangle commutes by the definition of $L_Q$ and the left-hand one and the inner square commute since $e_Q$ is an $F$-algebra homomorphism (see Remark~\ref{rem:aux}(2)). The upper part commutes since $\delta^\cst_X$ is the monoid multiplication $m_X$, see Remark~\ref{rem:aux}(3).
\end{proof}

\takeout{
\begin{example}
In $\Set$ the generators simply mean that every element of $M$ is a product of elements from $e_0[X]$. In $\Inv$ this means that every element of $M$ is a product of elements from $e_0[X]$ and their complements. In $\JSL$ an $X_0$-generated $\D$-monoid is an idempontent semiring with $e_0: X_0 \ra \under{M}$ such that every element of $M$ is a sum of products of elements from $e_0[X_0]$.
\end{example}
}

\section{Algebraic Recognition and Syntactic $\boldsymbol{\D}$-Monoids}
\label{sec:syn}

In classical algebraic automata theory one considers recognition of
languages by (ordinary) monoids in lieu of automata. The key concept
is the \emph{syntactic monoid} which is characterized as the smallest
monoid recognizing a given language. There are also related concepts of
canonical algebraic recognizers in the literature, e.g. the syntactic ordered monoid, the syntactic
idempotent semiring and the syntactic associative algebra. In this section we will give a uniform account of algebraic
language recognition in our categorical setting. Our main result is the
definition and construction of a minimal algebraic
recognizer, the \emph{syntactic $\D$-monoid} of a language. 

\begin{assumptions}
Throughout this section $(\D,\t,I)$ is an arbitrary symmetric monoidal closed category, and 
$\E$ is a class of epimorphisms in $\D$ that contains 
all isomorphisms and is closed under composition. In the case where $\D$ 
is a variety of algebras or ordered algebras, we always choose
\[ \E = \text{surjective homomorphisms}. \]
\end{assumptions}
\begin{definition}
\label{def:recog}
A $\D$-monoid morphism $e: \xs \ra M$ \emph{recognizes} the language $L: \xs \to Y$ if there exists a morphism $f: M\ra Y$ of $\D$ with $L=f\o e$.
\end{definition}

\begin{example}\label{ex:monrec}
We use the notation of Example \ref{ex:language}. 
\begin{enumerate}[label=(\arabic*)]
\item $\D = \Set$ with $\xs = X^*$ and $Y = \{ 0, 1\}$: given a monoid $M$, a 
function $f: M\ra\{0,1\}$ defines a subset $F=f^{-1}[1]\seq M$. Hence a monoid 
morphism $e: X^*\ra M$ recognizes $L$ via $f$ (i.e.~$L=f\o e$) iff $L_0 = 
e^{-1}[F]$. This is the classical notion of recognition of a language $L_0\seq 
X^*$ by a monoid, see e.g. \cite{pin15}.
\item $\D = \Pos$ with $\xs = X^*$ and $Y = \{ 0< 1\}$: given an ordered 
monoid $M$, a monotone map $f: M\ra\{0,1\}$ defines an upper set 
$F=f^{-1}[1]\seq M$. Hence a monoid morphism $e: X^*\ra M$ recognizes $L$ iff 
$L_0$ is the preimage of some upper set of $M$. This notion of recognition is 
due to Pin~\cite{pin95}.
\item $\D=\PSet$ with $X=X_0 + \{\bot\}$, $\xs = X_0^* + \{\bot\}$ and $Y=\{\bot,1\}$: given a monoid with zero $M$, a $\PSet$-morphism $f: M\ra \{\bot,1\}$ defines a subset $F=f^{-1}[1]$ of $M\setminus\{0\}$. A zero-preserving monoid morphism $e: X_0^*+\{\bot\}\ra M$ recognizes $L$ via $f$ iff $L_0 = e^{-1}[F]$.
\item $\D=\Inv$ with $X=X_0+\tl{X_0}$,  $\xs=X_0^* + \tl{X_0^*}$ and $Y=\{0,1\}$: for an involution monoid $M$ to give a morphism $f: M\ra \{0,1\}$ means to give a subset $F=f^{-1}[1]\seq M$ satisfying $m\in F$ iff $\tl{m}\not\in F$. Then $L$ is recognized by $e: X_0^*+\tl{X_0^*}\ra M$ via $f$ iff $L_0 = X_0^*\cap e^{-1}[F]$.
\item $\D=\JSL$ with $X=\Pow_f X_0$, $\xs = \Pow_f X_0^*$ and $Y=\{0,1\}$: for an idempotent semiring $M$ a morphism $f: M\ra Y$ defines  a prime upset $F=f^{-1}[1]$, see Example \ref{ex:automata}. Hence $L$ is recognized by a semiring homomorphism $e: \Pow_f X_0^* \ra M$ via $f$ iff $L_0 = X_0^*\cap e^{-1}[F]$. Here we identify $X_0^*$ with the set of all singleton languages $\{w\}$, $w\in X_0^*$. This is the concept of language recognition introduced by Pol\'ak \cite{polak01} (except that he puts $F=f^{-1}[0]$, so $0$ and $1$ must be swapped, as well as $F$ and $M\setminus F$).
\item $\D=\SMod{\S}$ with $X=\Psi X_0$, $\xs = \S[X_0]$ and $Y=\S$: given an associative algebra $M$, a language $L$ is recognized by $e: \S[X_0]\ra M$ via $f: M\ra \S$ iff $L= f\o e$. For the case where the semiring $\S$ is a ring, this notion of recognition is due to Reutenauer \cite{reu80}. 
\end{enumerate}
\end{example}

\begin{remark}\label{rem:Lalg}
\begin{enumerate}[label=(\arabic*)]
\item By an \emph{$X$-generated $\D$-monoid} we mean a morphism 
$e: \xs \epito M$ in $\DMon$ with $e\in \E$. Given two $X$-generated $\D$-monoids $e_i: \xs \epito M_i$, $i 
= 1,2$, we say, as usual, that $e_1$ is \emph{smaller or equal to} $e_2$ 
(notation: $e_1 \leq e_2$) if $e_1$ factorizes through $e_2$. Note that if 
$\D$ is a variety and 
$X=\Psi X_0$, the free $\D$-monoid $\xs=\Psi X_0^*$ on $X$ is also the free 
$\D$-monoid on the set $X_0$ (w.r.t. the forgetful functor $\DMon\ra\Set$), 
see Proposition~\ref{prop:freemon2}. In this case, to give a quotient 
$e:\xs\epito M$ is equivalent to giving an 
 $X_0$-indexed family $(m_x)_{x\in X_0}$ of generators for the $\D$-monoid $M$ -- which is why 
 $M$ may also be called an \emph{$X_0$-generated $\D$-monoid}.
\item   Let $e: \xs \epito M$ be an $X$-generated $\D$-monoid with unit $i: I \to M$ and multiplication $m: M \t M \to M$. Recall that $\eta_X: X \to \xs$ denotes the universal morphism of the free $\D$-monoid on $X$ and consider the $F$-algebra associated to $M$ and $X \xrightarrow{\eta_X} \xs \xrightarrow{e} M$ (see Definition~\ref{def:asso}). Thus, together with a given $f: M \to Y$ an $X$-generated $\D$-monoid induces an automaton $(M, \delta, i, f)$ called the \emph{derived automaton}.
\end{enumerate}
\end{remark}

\begin{lemma}
\label{lem:recog}
The language recognized by an $X$-generated $\D$-monoid $e: \xs \epito M$ via $f: M\ra Y$ is the language accepted by its derived automaton.
\end{lemma}

\begin{proof}
 By definition $e$ recognizes via $f$ the language $L = f \o e$. We are done once we prove that $e$ is the unique $F$-algebra homomorphism from $\xs$ to the $F$-algebra associated to $M$ and $e \o \eta_X$ (cf.~Remark~\ref{rem:Lalg}). Recall from Proposition~\ref{prop:initalg} that the initial $F$-algebra is the $F$-algebra associated to the free $\D$-monoid $\xs$ and $\eta_X$. Then the following diagram clearly commutes, since $e$ is a $\D$-monoid morphism:
  \[
  \xymatrix@C+3pc{
    F\xs = I + X \t \xs 
    \ar[r]^-{I + \eta_X \t \xs}
    \ar[d]_{Fe = I + X \t e}
    &
    I + \xs \t \xs 
    \ar[r]^-{[i_X,m_X]}
    \ar[d]_{I + e \t e}
    &
    \xs
    \ar[d]^e 
    \\
    FM = I + X \t M
    \ar[r]_-{I + (e\o \eta_X) \t M}
    &
    I + M \t M
    \ar[r]_-{[i,m]}
    &
    M
    }
  \]
  This completes the proof.
\end{proof}

We are now ready to give an abstract account of syntactic monoids in our setting. In classical algebraic automata theory the syntactic monoid of a language is characterized as the smallest monoid recognizing that language. We will use this property as our definition of the syntactic $\D$-monoid.
\begin{definition}
  \label{def:syn}
  The \emph{syntactic $\D$-monoid} of a language $L:\xs\ra Y$, denoted by $\Syn L$, is the smallest $X$-generated monoid recognizing $L$. 
\end{definition}
In more detail, the syntactic $\D$-monoid of $L$ is an $X$-generated $\D$-monoid $e_L: \xs \epito \Syn L$ together with a morphism $f_L: \Syn{L}\ra Y$ of $\D$ such that (i) $e_L$ recognizes $L$ via $f_L$, and (ii) for every $X$-generated $\D$-monoid $e: \xs \epito M$ recognizing $L$ via $f: M\ra Y$ we have $e_L \leq e$, that is, the left-hand triangle below commutes for some $\D$-monoid morphism $h$:
\[
\xymatrix{
  \xs
  \ar@{->>}[rr]^-e
  \ar@{->>}[rrd]_-{e_L}
  &&
  M
  \ar@{->>}[d]^h
  \ar[rr]^-f 
  &&
  Y
  \\
  &&
  \Syn L
  \ar[rru]_{f_L}
}
\]
Note that the right-hand triangle also commutes since $e$ is epimorphic and $f \o e = L = f_L \o e_L$. The universal property determines $\Syn L$, $e_L$ and $f_L$ uniquely up to isomorphism.

The above definition leaves open the question whether the syntactic $\D$-monoid of a 
language actually exists. In Section \ref{sec:existsynmon} we investigate the 
existence of syntactic $\D$-monoids in an abstract symmetric monoidal closed category $\D$. In 
Section 
\ref{sec:constsynmon} we show how to construct them in our 
algebraic setting, using the syntactic congruence of a language.

\subsection{Existence of syntactic $\boldsymbol{\D}$-monoids}\label{sec:existsynmon}

\begin{definition}
The category $\D$ is  \emph{stable} w.r.t. $\E$ if the following square is a pushout for all morphisms $a: A\epito A'$ and $b: B \epito B'$ in $\E$:
\[
\xymatrix{
& A\t B \ar[dr]^{A\t b} \ar[dl]_{a\t B} \ar[dd]^{a\t b}&\\
A'\t B \ar[dr]_{A'\t b} && A\t B' \ar[dl]^{a\t B}\\
& A' \t B'&
}
\]
\end{definition}

\begin{example}\label{ex:stable}
\begin{enumerate}[label=(\arabic*)]
\item $\Set$ is stable w.r.t.~$\E =$ surjective maps. The pushout of the surjections $a\times B$ and $A\times b$ is given by $A\times B/\mathord{\sim}$ where $\sim$ is the least equivalence relation such that
\[ a(x)=a(x') \text{ implies } (x,y)\sim (x',y) \text{ for all }  y\in B\]
as well as 
\[ b(y)=b(y') \text{ implies } (x,y)\sim (x,y') \text{ for all }  x\in A. \]
Obviously, $\sim$ is the kernel equivalence of $a\times b$.
\item More generally, every variety $\D$ of algebras is stable w.r.t. $\E =$ surjective homomorphisms. To see this, recall from Definition \ref{def:tensorproduct} that the forgetful functor from $\D$ to $\Set$
yields the tensor product $A\t B$ via the universal bimorphism $\eta_{A,B}: \und{A}\times\und{B}\to \und{A\t B}$. Suppose that $u: A'\t B \to C$ and $v: A\t B'\to C$ are given with $u\o (a\t B) = v \o (A\t b)$. Consider the following diagram:
\[
\xymatrix{
&& \und{A\t B} \ar[ddll]_{a\t B} \ar[ddrr]^{A\t b} &&\\
&& \und{A}\times\und{B} \ar[dl]_{a\times B} \ar[dd]^{a\times b} \ar[u]_\eta \ar[dr]^{A\times b} &&\\
\und{A'\t B} \ar[ddrr]_{A'\t b} \ar@/_1em/[dddrr]_u & \und{A'}\times\und{B} \ar[l]_\eta \ar[dr]_{A'\times b} && \und{A}\times\und{B'} \ar[r]^\eta \ar[dl]^{a\times B'} & \und{A\t B'} \ar[ddll]^{a\t B'} \ar@/^1em/[dddll]^v\\
&& \und{A'}\times\und{B'} \ar@/^3em/@{-->}[dd]^{w_0} \ar[d]_\eta &&\\
&& \und{A'\t B'} \ar@{-->}[d]^w&&\\
&& \und{C} &&
}
\]
By Example \ref{ex:stable}(1) above the inner square is a pushout in $\Set$, so there is a unique function $w_0: \und{A'}\times\und{B'}\to \und{C}$ with $u\o \eta_{A',B} = w_0\o (A'\times b)$ and $v\o \eta_{A,B'} = w_0 \o (a\times B')$. We claim that $w_0$ is a bimorphism of $\D$. Indeed, for each $x\in \und{A'}$ the map $w_0(x,-)$ carries a morphism of $\D$: choosing $\overline x\in \und{A}$ with $a(\overline x) =  x$, we have
\[ w_0(x,-) = w_0\o (a\times B')(\overline x, -) = v\o \eta_{A,B'}(\ol x,-) \]
and the last map is a morphism in $\D$ because $\eta_{A,B'}$ is a bimorphism. Symmetrically, $w_0(-,y)$ is a morphism of $\D$ for all $y\in \und{B'}$.

Consequently, the bimorphism $w_0$ induces a unique morphism $w: A'\t B'\to C$ in $\D$ with $w\o \eta_{A',B'} = w_0$. We have
\[  w\o (A'\t b) = u \]
because, by definition of $w_0$, this holds when precomposed with $\eta_{A',B}$. Analogously,
\[ w\o (a\t B') = v.\]
The uniqueness of $w$ follows from the uniqueness of $w_0$.
\item $\Pos$ is stable w.r.t. $\E =$ surjective monotone maps. The pushout of $a\times B$ and $A\times b$ is obtained from the pushout in $\Set$, i.e. $\und{A'}\times\und{B'}$, by taking the smallest preorder $\leq$ such that $(x,y)\leq (x',y)$ for $x\leq x'$ in $A'$ as well as $(x,y)\leq (x,y')$ for all $y\leq y'$ in $B'$ (and forming the corresponding quotient poset of the preordered set $A'\times B'$, cf. Remark \ref{rem:orderedalg} below). But this is just the product order; thus the pushout is the product $A'\times B'$ in $\Pos$ with the cocone $A'\times b$ and $a\times B'$.
\item More generally, every variety $\D$ of ordered algebras is stable w.r.t. 
$\E =$ surjective homomorphisms. The argument is completely analogous to 
Example \ref{ex:stable}(2).
\end{enumerate}
\end{example}

\begin{theorem}\label{thm:creatpush}
Let $\D$ be stable w.r.t. $\E$, and suppose moreover that $\E$ is closed under wide pushouts and 
tensor products. Then the forgetful functor from $\Mon{D}$ to $\D$ creates 
wide pushouts of $\E$-carried morphisms.
\end{theorem}

\begin{proof}
\begin{enumerate}[label=(\arabic*)]
\item Given $e_i: D\to D_i$ ($i\in I$) in $\E$ with wide pushout $e = \overline{e_i}\o e_i: D\to \overline D$ in $\D$, we prove that the morphisms $e_i\t e_i$ have the wide pushout $e\t e = (\ol{e_i}\t \ol{e_i}) \o (e_i\t e_i)$ in $\D$. Indeed, given a compatible family
\[ c: D\t D \to C \text{ and } c_i: D_i\t D_i\to C \text{ with } c = c_i\o (e_i\t e_i), \]
the morphisms $e_i\t D$, $i\in I$, have the compatible family formed by $c$ and $c_i\o (D_i\t e_i)$. Since $-\t D$ is a left adjoint and thus preserves colimits, the wide pushout of $e_i\t D$ is formed by $e\t D$ and $\ol{e_i}\t D$. Consequently there is unique morphism
\[  \tilde{c}: \ol{D}\t D\to D \]
such that the following squares commute for all $i\in I$:
\[
\xymatrix{
D_i\t D \ar[r]^{\ol{e_i}\t D} \ar[d]_{D_i\t e_i} & \ol{D}\t D \ar[d]^{\tl c} \\
D_i\t D_i \ar[r]_{c_i} & C 
}
\]
Since by assumption $\ol{e_i}\in \E$, the stability of $\D$ gives a unique morphism $\tl{c_i}$ making the following diagram commutative:
\[ 
\xymatrix{
& D_i\t D \ar[dl]_{D_i\t e_i} \ar[dr]^{\ol{e_i}\t D} & \\
D_i\t D_i \ar[dr]^{\ol{e_i}\t D_i} \ar[ddr]_{c_i} & & \ol{D}\t D \ar[dl]_{\ol{D}\t e_i} \ar[ddl]^{\tl c} \\
& \ol{D}\t D_i \ar@{-->}[d]^{\tl{c_i}} & \\
& C & 
}
\]
Thus $\tl{c}$ and $\tl{c_i}$ form a cocone of the family $\ol{D}\t e_i$. Since the latter family has the wide pushout formed by $\ol{D} \t e$ and $\ol{D}\t\ol{e_i}$ (using that the left adjoint $\ol{D}\t -$ preserves colimits), we obtain a unique morphism $\ol{c}: \ol{D}\t\ol{D}\to C$ for which the following triangle commutes:
\[  
\xymatrix{
\ol{D}\t D_i \ar[r]^{\ol{D}\t \ol{e_i}} \ar[dr]_{\tilde{c_i}} & \ol{D}\t \ol{D} \ar@{-->}[d]^{\ol{c}} \\
& C
}
\]
Combined with the definition of $\tilde{c_i}$, we get
\[ c_i = \tilde{c_i}\o (\ol{e_i}\t D_i) =  \ol{c}\o (\ol{D}\t \ol{e_i}) \o (\ol{e_i}\t D_i) = \ol{c}\o (\ol{e_i}\t  \ol{e_i}) \]
for all $i\in I$, as desired.

The uniqueness of $\ol c$ follows from the fact that the morphisms $\ol{e_i}\t\ol{e_i}$ are epimorphisms (using that $\ol D \t -$ and $-\t D_i$ are left-adjoints and hence preserve epimorphisms).
 
\item Suppose that $\D$-monoid structures are given such that the above morphisms $e_i$ are $\D$-monoid morphisms
\[ e_i: (D,\mu,\eta)\to (D_i,\mu_i,\eta_i) \]
for all $i\in I$. We define a $\D$-monoid structure $(\ol{D},\ol{\mu},\ol{\eta})$ such that all $\ol{e_i}$ are $\D$-monoid morphisms. For the unit $\ol{\eta}$ put
\[ \ol{\eta}~\equiv~ I \xra{\eta} D\xra{{e}} \ol{D}. \]
Next, observe that the morphisms $e_i\t e_i$ have a compatible cocone consisting of $e\o \mu$ and $\ol{e_i}\o \mu_i$ ($i\in I$). Hence by (1) there exists a unique $\ol{\mu}: \ol{D}\t\ol{D}\to\ol{D}$ for which the following squares commute:
\[ 
\xymatrix{
D\t D \ar[r]^\mu \ar[d]_{e\t e} & D \ar[d]^e && D_i\t D_i \ar[r]^{\mu_i} \ar[d]_{\ol{e_i}\t\ol{e_i}} & D_i \ar[d]^{\ol{e_i}}\\
\ol{D}\t \ol{D} \ar[r]_{\ol{\mu}} & \ol{D} && \ol{D}\t\ol{D} \ar[r]_{\ol{\mu}} & \ol{D}
}
\]
We only need to prove that $(\ol{D},\ol{\mu},\ol{\eta})$ is a $\D$-monoid. Then the definitions of $\mu$ and $\eta$ immediately imply that $e$ and $\ol{e_i}$ are $\D$-monoid morphisms, and the verification that they form the wide pushout in $\Mon{D}$ is trivial.

\emph{Unit laws.} Due to symmetry we only verify the left unit law $\ol{\mu} \o (\ol{\eta}\t\ol{D}) = \lambda_{\ol{D}}$, where $\lambda_{\ol{D}}: I\t \ol{D} \cong \ol{D}$ is the left unit isomorphism. Consider the following diagram:
\[
\xymatrix{
I\t D \ar[rrr]^{\eta\t D} \ar[dr]^{I\t e} \ar[dddrr]_{\lambda_D}  & & & D\t D \ar[dl]_{e\t e} \ar[dddl]^\mu \\
& I\t \ol{D} \ar[r]^{\ol{\eta}\t\ol{D}} \ar[dr]_{\lambda_{\ol D}} & \ol{D}\t \ol{D} \ar[d]^{\ol{\mu}} & \\
& & \ol{D} & \\
& & D \ar[u]_{{e}} &
}
\]
The outside triangle is the unit law of $(D,\mu,\eta)$. The upper part commutes due to $\ol{\eta} = e\o \eta$, the left-hand one is the naturality of $\lambda$, and the right-hand one is the definition of $\ol{\mu}$.  Thus the inner triangle commutes when precomposed with $I\t e$. Since the latter is an epimorphism, the proof is complete.

\emph{Associative law.} Consider the following diagram, where $\alpha$ denotes the associativity isomorphism:
\[
\xymatrix{
 (D\t D)\t D \ar[rr]^\alpha \ar[ddd]_{\mu\t I} \ar[dr]^{( e\t  e)\t  e} & & D\t (D\t D) \ar[rr]^{D\t \mu} \ar[d]^{ e\t ( e \t  e)}& & D\t D \ar[dl]_{e\t  e} \ar[ddd]^\mu \\
 & (\ol D\t \ol D) \t \ol D \ar[r]^-\alpha \ar[d]_{\ol\mu\t \ol D} & \ol D\t (\ol D\t \ol D) \ar[r]^-{\ol D\t \ol \mu} & \ol D\t \ol D \ar[d]^{\ol\mu} & \\
 & \ol D\t \ol D \ar[rr]_{\ol\mu} & &  \ol D & \\
D\t D \ar[ur]_{ e\t  e} \ar[rrrr]_\mu & & & & D \ar[ul]_{ e}
}
\]
The outside is the associativity of $(D,\mu,\eta)$. All the inner parts 
except for the (desired) inner square commute: for the upper left-hand one use 
the naturality of $\alpha$, and all other parts follow from the definition of 
$\ol\mu$. Consequently the inner square commutes, as it commutes when 
precomposed with the epimorphism $(e\t e)\t e$.
\end{enumerate}
\end{proof}

\begin{corollary}\label{cor:syntmonexists}
Let $\D$ satisfy the assumptions of the previous theorem, and suppose moreover that  
$\D$ has wide 
pushouts of $\E$-morphisms. Then every language $L:\xs\to Y$ has a syntactic 
$\D$-monoid.
\end{corollary}

\begin{proof}
Let $e_i: \xs \epito M_i$ ($i\in I$) be the family of all $X$-generated
$\D$-monoids recognizing $L$. Note that $I\neq\emptyset$ since the identity 
morphism of $\xs$ trivially recognizes $L$. Form the wide pushout $e: 
\xs\epito 
M$ and $\ol{e_i}: 
M_i\epito M$  of all $e_i$'s in $\D$. By Theorem \ref{thm:creatpush} there 
is a unique 
$\D$-monoid structure on $M$ making $e$ and $\ol{e_i}$ a wide pushout in 
$\Mon{\D}$. We claim that $e: \xs\epito M$ is the syntactic $\D$-monoid of 
$L$. To this end we verify the universal property of Definition \ref{def:syn}:

(a) Since $e_i$ recognizes $L$, there exists a morphism $f_i: M_i\to Y$ in 
$\D$ with $L=f_i\o e_i$ for each $i\in I$. Hence the morphisms $L$ and $f_i$ 
form a compatible family, so there is a unique morphism $f: M\to Y$ in $\D$ 
with $f_i = f\o \ol{e_i}$ for all $i$. Choosing an arbitrary $i\in I$, it follows that
\[ f\o e = f\o \ol{e_i}\o e_i = f_i\o e_i = L, \]
i.e. $e$ recognizes $L$ via $f$.

(b) By construction, for every $X$-generated $\D$-monoid $e_i: 
\xs\epito 
M_i$ recognizing $L$ we have the $\D$-monoid morphism $\ol{e_i}: 
M_i\epito 
M$ with $\ol{e_i}\o e_i = e$.
\end{proof}

\begin{corollary}
Every language $L\colon \xs\to Y$ in a variety $\D$ of algebras or ordered algebras admits a syntactic $\D$-monoid.
\end{corollary}
We shall now show that in the situation of the above corollary the syntactic $\D$-monoid of a language $L$ admits a more concrete construction via the \emph{syntactic congruence} of $L$. We first consider the case of varieties of algebras (Section \ref{sec:constsynmon}) and then turn to varieties of ordered algebras (Section \ref{sec:constsynmonord}).

\subsection{Syntactic congruences for varieties of algebras}\label{sec:constsynmon}
In this subsection, let $\D$ be a variety of algebras.

\begin{remark}\label{rem:congalg}
\begin{enumerate}[label=(\arabic*)]
\item Recall that a \emph{congruence} on a $\Sigma$-algebra $A$ is an 
equivalence 
relation $\equiv$ on $A$ that forms a subalgebra of $A\times A$. We 
denote by $A/\mathord{\equiv}$ the quotient algebra modulo~$\equiv$.
\item For any homomorphism 
$h: A\to B$ of $\Sigma$-algebras, the \emph{kernel congruence} of $h$ is the congruence on $A$ 
defined by
\[ a \equiv_h a' \quad\text{iff}\quad h(a)= h(a'). \]
\item We will frequently use the following \emph{homomorphism theorem}: given 
homomorphisms of $\Sigma$-algebras $e: A\epito B$ and $h: A\to C$, where $e$ 
is surjective, there exists a homomorphism $h': B\to C$ with $h = h'\o e$ iff, 
for all $a,a'\in A$, \[e(a)=e(a') \quad\text{implies}\quad h(a)=h(a').\]
\end{enumerate}
\end{remark}

\begin{definition}\label{def:syncong}
Let $\D$ be a variety of algebras. The \emph{syntactic congruence} of a 
language $L: \xs\to Y$ is the relation $\equiv_L$ on $\under{\xs}$ defined by
\begin{equation}\label{eq:syncong}
 u\equiv_L v \quad\text{iff}\quad \forall x,y\in\under{\xs}: L(x\bullet u\bullet y) = 
 L(x\bullet v\bullet y).
 \end{equation}
\end{definition}

\begin{theorem}\label{prop:synmon}
 $\equiv_L$ is a congruence on the free $\D$-monoid $\xs$, and
\[ \Syn{L} = \xs /\mathord{\equiv_L}. \]
\end{theorem}

\begin{proof}
(a) Clearly $\equiv_L$ forms an equivalence relation on $\xs$. To show that it 
is  congruence, first observe that $\equiv_L$ is a subobject of $\xs\times\xs$ 
in $\D$. Indeed, we have $\mathord{\equiv_L} = \bigcap K_{x,y}$, where for 
fixed $x,y\in\xs$ the object $K_{x,y}$ is the kernel of the $\D$-morphism
     $\xymatrix@1{\xs \ar[r]^-{x \bullet -} & \xs \ar[r]^-{-\bullet y} & \xs 
     \ar[r]^-L & Y}$. Moreover, $\equiv_L$ is closed under the monoid 
     multiplication of $\xs \times \xs$: given $u \equiv_L v$ and $u'\equiv_L 
     v'$ we have for all $x, y \in \xs$ that
\begin{align*}
L(x \bullet u \bullet u' \bullet y) &=  L(x \bullet v \bullet u' \bullet y) & (\text{$y:= u'\bullet y$ in \eqref{eq:syncong}})\\
&= L(x \bullet v \bullet v' \bullet y) & (\text{$x:= x\bullet v$ in \eqref{eq:syncong}}).
\end{align*}
 Hence $u \bullet u' \equiv_L v \bullet v'$ and therefore 
$\equiv_L$ is a $\D$-submonoid of $\xs\times \xs$, i.e.~a congruence of $\xs$.
     
(b) Denote by $e_L: \xs\epito \xs /\mathord{\equiv_L}$ the projection. Then 
$e_L(u) = e_L(v)$ (i.e.~$u\equiv_L v$) implies $L(u)=L(v)$ by putting $x=y=1$ 
in \eqref{eq:syncong}, where $1$ is the unit of the monoid $\xs$. Hence the homomorphism theorem yields a morphism $f_L: 
\xs /\mathord{\equiv_L}\to Y$ in $\D$ with $f_L\o e_L = L$, i.e.~$e_L$ 
recognizes $L$ via $f_L$.

(c) In remains to verify the universal property of Definition 3.5. Let $e: 
\xs\epito M$ be a surjective $\D$-monoid morphism recognizing $L$ via $f: 
M\to Y$. To construct the morphism $h: M\to \xs /\mathord{\equiv_L}$, we again 
use the homomorphism theorem. Thus let $u,v\in\xs$ with $e(u)=e(v)$. Then, for 
all $x,y\in\xs$,
\begin{align*}
L(x\bullet u \bullet y) &= f(e(x\bullet u\bullet y)) & \text{($L=f\o e$)}\\
&= f(e(x)e(u)e(y)) & \text{($e$ preserves $\bullet$)}\\
&= f(e(x)e(v)e(y)) & \text{($e(u)=e(v)$)}\\
&= f(e(x\bullet v\bullet y)) & \text{($e$ preserves $\bullet$)}\\
&= L(x\bullet v \bullet y) & \text{($L=f\o e$)}
\end{align*}
Hence $u\equiv_L v$, i.e.~$e_L(u) = e_L(v)$. Thus the homomorphism theorem 
yields the desired $\D$-monoid morphism $h: M\to \xs /\mathord{\equiv_L}$ with 
$h\o e = 
e_L$.
\end{proof}

\begin{example} Using the notation of Example $\ref{ex:language}$ we obtain 
the following  syntactic algebras:
  \begin{enumerate}[label=(\arabic*)]
  \item In $\Set$ with $Y = \{0,1\}$, the syntactic monoid of a language 
  $L\seq X^*$ is the quotient monoid $X^*/ \mathord{\equiv_L}$, where for  
  $u,v\in X^*$,
  \[
  u \equiv_L v \quad\text{iff}\quad  \text{for all } 
  x, y \in X^*: xuy \in L \iff xvy \in L.
  \]
  This construction is due to Rabin and Scott \cite{rs59}.
  \item In $\PSet$ with $X=X_0+\{\bot\}$ and $Y=\{\bot,1\}$ the 
  \emph{syntactic monoid with zero} of a language $L_0\seq X_0^*$ is 
  $(X_0^*+\{\bot\}) /\mathord{\equiv_L}$ where for  $u,v\in 
  X_0^*+\{\bot\}$,
  \[ u\equiv_L v \quad\text{iff}\quad \text{for all $x,y\in X_0^*$}: xuy\in 
  L_0 
  \iff xvy \in L_0.\]
 The zero element is the congruence class of $\bot$.
  \item In $\Inv$ with $X = X_0 + \wt{X_0}$ and $Y = \{0,1\}$ the 
  \emph{syntactic involution monoid} of a language $L_0\seq X_0^*$ is the 
  quotient of $X_0 + \wt{X_0^*}$ modulo the congruence $\equiv_L$ defined for 
  words $u,v \in X_0^*$ as follows:
  \begin{enumerate}[label=(\roman*)]
  \item $u \equiv_L v\quad\text{iff}\quad\wt u \equiv_L \wt 
  v\quad\text{iff}\quad\text{for all $x, y \in X_0^*$}: xuy \in L_0 \iff xvy 
  \in L_0$;
\item
$u \equiv_L \wt v\quad\text{iff}\quad\wt u \equiv_L v 
\quad\text{iff}\quad\text{for 
all $x, y \in X_0^*$}: xuy \in L_0 \iff xvy \not\in L_0$.
\end{enumerate}
\item In $\SMod{\S}$  with $X = \Psi X_0$ and $Y = \S$ the \emph{syntactic 
associative $\S$-algebra} of a weighted language $L_0: X_0^*\ra\S$ is the 
quotient of $\S[X_0]$ modulo the congruence defined for $U,V \in \S[X_0]$ as 
follows: 
\begin{equation}
  \label{eq:syns}
  U \equiv_L V\quad \text{iff} \quad \text{for all $x, y \in X_0^*$}: L(xUy) = 
  L(xVy)
\end{equation}
 Indeed, since $L:\S[X_0]\ra \S$ is linear, \refeq{eq:syns} implies $L(PUQ) = 
 L(PVQ)$ for all $P, Q \in \S[X_0]$, which is the syntactic congruence of 
 Definition \ref{def:syncong}.
\item In particular, for $\D = \JSL$ with $X=\Pow_f X_0$ and $Y=\{0,1\}$, we 
get  the \emph{syntactic (idempotent) semiring} of a language $L_0 \subseteq 
X_0^*$ introduced by Pol\'ak \cite{polak01}: it is the quotient $\Pow_f 
X_0^*/\mathord{\equiv_L}$ where for $U,V\in\Pow_f X_0^*$ we have 
\[
U \equiv_L V\quad\text{iff}\quad\text{for all $x, y \in X_0^*$}: (xUy) \cap 
L_0 
\neq \emptyset \iff xVy \cap L_0 \neq \emptyset.
\]
\item For $\D = \Vect{\K}$ with $X=\Psi X_0$ and $Y=\K$, the \emph{syntactic 
$\K$-algebra} of a $\K$-weighted language $L_0: X_0^*\ra \K$ is the quotient 
$\K[X_0]/I$ of the $\K$-algebra of finite weighted languages modulo the ideal
\[
I = \{ V \in \K[X_0] \mid \text{for all $x,y \in X_0^*$}: L(xVy) = 0 \}.
\]
 Indeed, the congruence this ideal $I$ generates ($U\equiv_L V$ iff $U-V \in 
 I$) 
 is precisely~\refeq{eq:syns}. Syntactic $\K$-algebras were studied by 
 Reutenauer~\cite{reu80}.
\item Analogously, for $\D = \Ab$ with $X=\Psi X_0$ and $Y=\Int$, the 
\emph{syntactic ring} of a $\Int$-weighted language $L_0: X_0^*\ra \Int$ is 
the quotient $\Int[X_0]/I$, where $I$ is the ideal of all $V\in \Int[X_0]$ 
with $L(xVy) = 0$ for all $x, y \in X_0^*$.
\end{enumerate}
\end{example}

\subsection{Syntactic congruences for ordered algebras}\label{sec:constsynmonord}
Next, we consider the case where $\D$ is a variety of ordered algebras.
\begin{remark}\label{rem:orderedalg}
\begin{enumerate}[label=(\arabic*)]
\item A \emph{congruence} on an ordered $\Sigma$-algebra $(A,\leq)$ is a 
preorder $\preccurlyeq$ on $A$ such that $a\leq a'$ implies $a\preccurlyeq 
a'$, and all $\Sigma$-operations of $A$ are monotone w.r.t. 
$\preccurlyeq$ (equivalently, $\preccurlyeq$ forms a subalgebra of 
$A\times A$). Then $\preccurlyeq \cap \succcurlyeq$ is a 
congruence in the unordered sense, and the quotient algebra $e: A\epito 
A/(\preccurlyeq \cap \succcurlyeq)$ forms an ordered algebra with partial 
order induced by $e$:
\[ e(u) \preccurlyeq' e(v) \quad\text{iff}\quad u\preccurlyeq v.\]
This ordered algebra is called the \emph{quotient} of $A$ modulo 
$\preccurlyeq$,  and is denoted by $A/\mathord{\preccurlyeq}$.
\item For any 
homomorphism $h: A\to B$ between ordered $\Sigma$-algebras, the \emph{kernel congruence} of $h$ is the congruence on $A$ defined by
\[ a\leq_h a' \quad\text{iff}\quad h(a)\leq h(a'). \]
\item The \emph{homomorphism theorem} has the following version for ordered 
algebras, see e.g. \cite[Proposition 1.1]{pinweil}: given
homomorphisms of ordered $\Sigma$-algebras $e: A\epito B$ and $h: A\to C$, 
where 
$e$ 
is surjective, there exists a homomorphism $h': B\to C$ with $h = h'\o e$ iff, 
for all $a,a'\in A$, 
 \[e(a)\leq e(a') \quad\text{implies}\quad h(a)\leq h(a').\]
\end{enumerate}
\end{remark}

\begin{definition}
Let $\D$ be a variety of ordered algebras. The \emph{syntactic congruence} of 
a language $L: \xs\to Y$ is the relation $\leq_L$ on $\under{\xs}$ defined by
\begin{equation}\label{eq:syncongord}
 u\leq_L v \quad\text{iff}\quad \forall x,y\in\under{\xs}: L(x\bullet u\bullet y) \leq 
 L(x\bullet v\bullet y)
 \end{equation}
\end{definition}

\begin{theorem}
$\leq_L$ is a congruence of the free $\D$-monoid $\xs$, and
\[ \Syn{L} = \xs / \mathord{\leq_L}.\]
\end{theorem}

\begin{proof}
Clearly $\leq_L$ is a preorder. Let us verify the conditions of
Remark \ref{rem:orderedalg}(1). First, if $u\leq v$ in $\xs$, then $x\bullet 
u\bullet y \leq x\bullet v\bullet y$ for all $x,y\in\xs$ because $\bullet$ is 
monotone. Since also $L$ is monotone, we conclude $L(x\bullet u\bullet y) \leq 
L(x\bullet v\bullet y)$ in $Y$ for all $x,y$, i.e.~$u\leq_L v$. Next, observe 
 that $\leq_L$ is a subobject of $\xs\times\xs$ in $\D$, namely 
$\mathord{\leq_L} = \bigcap \leq_{x,y}$, where for fixed $x,y\in\xs$ the 
object 
$\leq_{x,y}$ is the kernel congruence of the $\D$-morphism
     $\xymatrix@1{\xs \ar[r]^-{x \bullet -} & \xs \ar[r]^-{-\bullet y} & \xs 
     \ar[r]^-L & Y}$. Moreover, the monoid multiplication $\bullet$ is 
     monotone w.r.t. $\leq_L$: given $u \leq_L v$ and $u'\leq_L v'$ we have for 
     all $x, y \in \xs$ that
\[
L(x \bullet u \bullet u' \bullet y) \leq L(x \bullet v \bullet u' \bullet y) 
\leq L(x \bullet v \bullet v' \bullet y).
\]
 Hence $u\bullet u' \leq_L v\bullet v'$, and 
therefore $\leq_L$ is a $\D$-submonoid of $\xs\times \xs$, as required.
     
 The proof that $\Syn{L} = \xs / \mathord{\leq_L}$ is completely analogous to 
 that of Proposition \ref{prop:synmon}: replace equations by inequations, 
 and use the homomorphism theorem for ordered algebras to construct the 
 morphisms $f_L$ and $h$.
\end{proof}

\begin{example} In $\D = \Pos$ with $Y = \{0<1\}$, the \emph{syntactic ordered 
monoid} of a language $L\seq X^*$ is the ordered quotient monoid 
$X^*/\mathord{\leq_L}$ where for $u,v\in X^*$,
\[ u\leq_L v \quad\text{iff}\quad \text{for 
all 
$x,y\in X^*$: } xuy\in L \Ra xvy\in L. \]
This construction is due to Pin \cite{pin95}.
\end{example}

\section{Transition $\boldsymbol{\D}$-Monoids}
\label{sec:tran}

In this section we present another construction of the syntactic $\D$-monoid of a language: it is the transition $\D$-monoid of the minimal $\D$-automaton for this language. We continue to work under the Assumptions \ref{ass:global}. Recall from Remark~\ref{rem:concept} the $\D$-monoid $[Q,Q]$ and the $\D$-monoid morphism $(\lambda\delta)^+: \xs\to[Q,Q]$.
\begin{definition}
  \label{def:TA}
  The \emph{transition $\D$-monoid} $\T{Q}$ of an $F$-algebra $(Q,\delta, i)$ is the image of the $\D$-monoid morphism $(\lambda\delta)^+:\xs \ra [Q,Q]$:
  \[
  \xymatrix{
    \xs \ar@{->>}[dr]_{e_{\T Q}} \ar[rr]^{(\lambda\delta)^+} & & [Q,Q] \\
    & \T Q \ar@{>->}[ur]_{m_{\T Q}} 
  }
  \]
\end{definition}

\begin{example}\label{ex:42}
\begin{enumerate}[label=(\arabic*)]
\item In $\Set$ or $\Pos$ the \emph{(ordered) transition monoid} of an $F$-algebra $Q$, i.e.~an 
(ordered) automaton without final states, is the (ordered) monoid of all extended 
transition maps $\delta_w = \delta_{a_n}\o \ldots\o \delta_{a_1}: Q\ra Q$ for 
$w=a_1\cdots a_n\in X^*$. The unit is $\id_Q = \delta_\epsilon$ and the monoid multiplication is  composition.
\item In $\PSet$ with $X=X_0+\{\bot\}$ (the setting for partial automata) this is completely analogous, except that we add the constant endomap of $Q$ with value $\bot$.
\item In $\Inv$ with $X=X_0+\tl{X_0}$ we get the involution monoid of all $\delta_w$ and $\tl{\delta_w}$. Again the unit is $\delta_\epsilon$, and the multiplication is determined by composition plus the equations $x\tl y = \tl{xy} = \tl x y$.
\item In $\JSL$ with $X=\Pow_f X_0$ the \emph{transition semiring} consists of all finite joins of extended transitions, i.e.~all semilattice homomorphisms of the form
$\delta_{w_1}\vee\cdots\vee \delta_{w_n}$ for  $\{w_1,\ldots,w_n\}\in\Pow_f X_0^*$.
The transition semiring was introduced by Pol\'ak \cite{polak01}.
\item In $\SMod{\S}$ with $X=\Psi X_0$ the associative transition algebra  consists of all linear maps of the form
$\sum_{i=1}^n s_i \delta_{w_i}$ with $s_i\in \S$ and  $w_i\in X_0^*$.
\end{enumerate}
\end{example}

Recall from Definition~\ref{def:aut} that a $\D$-automaton is an $F$-algebra $Q$ together with an output morphism $f: Q \to Y$. Hence we can speak of the transition $\D$-monoid of a $\D$-automaton.
\begin{proposition}\label{prop:transmon}
The language accepted by a $\D$-automaton $(Q,\delta,f,i)$ is recognized by the $\D$-monoid morphism $e_{\T{Q}}: \xs\epito \T{Q}$.
\end{proposition}
\begin{proof}
Let $(Q,\delta, i, f)$ be a $\D$-automaton. By definition it accepts the language $L_Q = (\xs \xra{e_Q} Q \xra{f} Y)$ where $e_Q$ is the unique $F$-algebra homomorphism. Consider the morphism that evaluates any endomorphism of $Q$ at the initial state:
\[
\ev_i = ([Q,Q] \cong [Q,Q]\t I \xra{[Q,Q]\t i} [Q,Q]\t Q \xra{\ev} Q).
\]
Now let 
\[
f_{\T Q} = (\T Q \xra{m_{\T Q}} [Q,Q] \xra{\ev_i} Q \xra{f} Y).
\]
With this morphism $\T Q$ recognizes $L$; indeed, using the right unit isomorphism $\iota_Z: Z \to Z \t I$ we compute:
\begin{align*}
  L_Q & = f \o e_Q \\
    & = f \o \ev \o ((\lambda\delta)^+ \t i) \o \iota_{\xs} & \text{(Remark~\ref{rem:concept})}\\
    & = f \o \ev_i \o \iota_{[Q,Q]}^{-1}\o ((\lambda\delta)^+ \t I) \o \iota_{\xs} & \text{(def.~of $\ev_i$)} \\
    & = f \o \ev_i  \o (\lambda\delta)^+ & \text{(naturality of $\iota$)} \\
    & = f \o \ev_i \o m_{\T Q} \o e_{\T Q} & \text{(Definition~\ref{def:TA})} \\
    & = f_{\T Q} \o e_{\T Q} & \text{(def.~of $f_{\T Q}$)}
\end{align*}
This completes the proof.
\end{proof}

\begin{definition}\label{def:minaut}
A $\D$-automaton $(Q,\delta ,i, f)$ is called \emph{minimal} iff it is
\begin{enumerate}[label=(\alph*)]
\item \emph{reachable}: the unique $F$-algebra homomorphism $\xs\ra Q$ is surjective;
\item \emph{simple}: the unique $T$-coalgebra homomorphism $Q\ra [\xs, Y]$ is injective.
\end{enumerate}
\end{definition}

\begin{theorem}[Goguen \cite{goguen75}]\label{thm:minaut}
Every language $L: \xs \ra Y$ is accepted by a minimal $\D$-automaton $\Min{L}$, unique up to isomorphism. Given any reachable automaton $Q$ accepting $L$, there is a unique surjective automata homomorphism from $Q$ into $\Min{L}$.
\end{theorem}

This leads to the announced construction of syntactic $\D$-monoids via  transition $\D$-monoids. The case $\D=\Set$ and $\Pos$ are standard results of algebraic automata theory (see e.g. Pin \cite{pin15}), and the case $\D=\JSL$ is due to Pol\'ak \cite{polak01}. For the other instances in Example~\ref{ex:42} this appears to be new.

\begin{remark}
  Recall that in our algebraic setting the tensor unit $I$ is $\Psi 1$. Hence 
  elements $a \in A$ of a $\D$-object are in 1-1-correspondence with morphisms 
  $a: I \to A$; we shall henceforth not distinguish these. 
\end{remark}
\begin{theorem}\label{thm:tran}
The syntactic $\D$-monoid of a language $L: \xs\ra Y$ is isomorphic to the transition $\D$-monoid of its minimal $\D$-automaton:
\[ \Syn{L} \cong \T{\Min{L}}. \]
\end{theorem}
\begin{proof}
Let $\Min{L}=(Q,\delta,i,f)$, and write $\delta_x: Q\ra Q$ for $e_{\T{Q}}(x)$ 
($x\in\xs$). Note that $\delta_{x\bullet y}= \delta_y\o \delta_x$ for all 
$x,y\in \xs$ since $e_{\T{Q}}$ is a $\D$-monoid morphism. Recall also that 
the unique $F$-algebra homomorphism $e_Q: \xs\ra Q$ assigns to $x\in\xs$ the 
element $\delta_x\o i: I\ra Q$ (see Remark \ref{rem:concept}(2)), and the 
unique $T$-coalgebra homomorphism 
$m_Q: Q\ra [\xs, Y]$ assigns to a state $q: I\ra Q$ the language $x\mapsto 
f\o\delta_x\o q$ (see Remark \ref{rem:finalmorphism}). It suffices to show 
that the kernel congruence of 
$e_{\T{Q}}$ coincides with the syntactic congruence of $L$. If $\D$ is a variety of 
algebras, this requires to prove that for all $u,v\in \under{\xs}$ 
one has
\[
\delta_u = \delta_v \quad\text{iff}\quad \forall x,y\in \xs: L(x\bullet u \bullet y)=L(x\bullet v \bullet y).
\]
To see this, we compute
\begin{align*}
\delta_u = \delta_v &\Lra \forall x: \delta_u \o e_Q(x)= \delta_v\o e_Q(x) & \text{($e_Q$ surjective)}\\
&\Lra\forall x: \delta_u\o \delta_x\o i = \delta_v\o \delta_x \o i & \text{(def. $e_Q$)}\\
& \Lra\forall x: m_Q\o \delta_u\o \delta_x\o i = m_Q\o \delta_v\o \delta_x \o i & \text{($m_Q$ injective)}\\
& \Lra \forall x, y: f \o \delta_y\o \delta_u\o \delta_x\o i = f\o \delta_y\o \delta_v\o \delta_x\o i & \text{(def. $m_Q$)}\\
& \Lra \forall x, y: f\o \delta_{x\bullet u\bullet y}\o i = f \o \delta_{x\bullet v \bullet y} \o i & \text{(def. $\delta_{(\mathord{-})}$)}\\
& \Lra \forall x, y: f\o e_Q(x\bullet u\bullet y) = f \o e_Q(x\bullet v\bullet y) & \text{(def. $e_Q$)}\\
& \Lra \forall x, y: L(x\bullet u \bullet y) = L(x\bullet v \bullet y) & \text{($L=L_Q$)}
\end{align*}
The case where $\D$ is a variety of ordered algebras is completely analogous: 
just replace equations by inequations.
\end{proof}

\section{$\boldsymbol{\D}$-Regular Languages}
\label{sec:rat}
Our results so far apply to arbitrary languages in $\D$. In the present section we focus on \emph{regular languages}, which in $\D=\Set$ are the languages accepted by finite automata, or equivalently the languages recognized by finite monoids. For arbitrary $\D$ the role of finite sets is taken over by finitely presentable objects. Recall that an object $D$ of $\D$ is \emph{finitely presentable} if the hom-functor $\D(D,\mathord{-}):\D\ra\Set$ preserves filtered colimits. Equivalently, $D$ can be presented with finitely many generators and relations. For example, the tensor unit $I=\Psi 1$ is finitely presentable, since it is presented with one generator and no relations.

\begin{definition}
A language $L: \xs\ra Y$ is called \emph{$\D$-regular} if it is accepted by some $\D$-automaton with a finitely presentable object of states.
\end{definition}
To work with this definition, we need the following

\begin{assumptions}\label{asm:finite}
We assume that the full subcategory $\D_{f}$ of finitely presentable objects of $\D$ is closed under subobjects, quotients and finite products.
\end{assumptions}

\begin{example}\label{ex:dregassumptions}
\begin{enumerate}[label=(\arabic*)]
\item Recall that a variety is \emph{locally finite} if all finitely presentable algebras (equivalently all finitely generated free algebras) are finite. Every locally finite variety satisfies the above assumptions. This includes our examples $\Set$, $\PSet$, $\Inv$ and $\JSL$.
\item A semiring $\S$ is called \emph{Noetherian} if all submodules of finitely generated $\S$-modules are finitely generated. In this case, as shown in \cite{bms13}, the category $\SMod{\S}$ satisfies our assumptions.
Every field is Noetherian, as is every finitely generated commutative ring, so  $\Vect{\K}$ and $\Ab=\SMod{\Int}$ are special instances. 
\end{enumerate}
\end{example}

\begin{example}
  As we have mentioned already, for $\D = \Set$ the $\D$-regular languages are precisely the classical regular languages. The same is true for $\PSet$, $\Pos$, $\Inv$ and $\JSL$. For $\D = \SMod{\S}$ for a Noetherian semiring $\S$ it is known that the $\D$-regular languages are precisely the rational weighted languages, see e.g.~\cite{bms13}.
\end{example}

\begin{remark} 
\begin{enumerate}[label=(\arabic*)] 
\item The functor $FQ=I+X\t Q$ preserves surjective homomorphisms. Therefore the factorization system of $\D$ lifts to the category of $F$-algebras, that is, every $F$-algebra homomorphism factorizes into a surjective homomorphism followed by an injective (resp. order-reflecting) one. By the \emph{reachable part} $Q_r$ of a $\D$-automaton $(Q,\delta,f,i)$ we mean the image of the initial $F$-algebra homomorphism, i.e.  $\xymatrix{
e_Q = (\xs \ar@{->>}[r]^>>>>>{e_r} & Q_r \ar@{>->}[r]^{m_r} & Q). 
}$
Putting $f_r := f\o m_r: Q_r\ra Y$, the $F$-algebra $Q_r$ becomes an automaton, and $m_r$ an automata homomorphism.\qed
\item In the following automata, (co-)algebras and monoids with finitely presentable carrier are referred to as \emph{fp-automata}, \emph{fp-(co-)algebras} and \emph{fp-monoids}, respectively.
\end{enumerate}
\end{remark}

\begin{theorem}\label{thm:reg}
For any language $L: \xs \ra Y$ the following statements are equivalent:
\begin{enumerate}[label=(\alph*)]
\item $L$ is $\D$-regular.
\item The minimal $\D$-automaton $\Min{L}$ has a finitely presentable carrier.
\item $L$ is recognized by some $\D$-monoid with a finitely presentable carrier.
\item The syntactic $\D$-monoid $\Syn{L}$ has a finitely presentable carrier.
\end{enumerate}
\end{theorem}

\begin{proof}
 (a)$\Lra$(b) follows from Theorem \ref{thm:minaut} and our assumption that $\D_{f}$ is closed under subobjects and  quotients. Similarly, $(c)\Lra(d)$ follows from the universal property of the syntactic monoid (see Definition \ref{def:syn}) and again closure of $\D_{f}$ under subobjects and quotients. (c)$\Ra$(a) is a consequence of Lemma \ref{lem:recog}. To prove (a)$\Ra$(c), let $Q$ be any fp-automaton accepting $L$. Then by Proposition \ref{prop:transmon} the transition monoid $\T{Q}\monoto[Q,Q]$ recognizes $L$, so by closedness of $\D_{f}$ under subobjects it suffices to show that $[Q,Q]$ is a finitely presentable object of $\D$. Assuming that $Q$ has $n$ generators as an algebra of $\D$, the map $[Q,Q]\ra Q^n$ defined by restriction to the set of generators is an injective (resp. order-reflecting) morphism of $\D$. Since $\D_{f}$ is closed under subobjects and finite products, it follows that $[Q,Q]$ is finitely presentable.
\end{proof}

Just as the collection of all languages is internalized by the final coalgebra $\fc$, see Proposition \ref{prop:fincoalg}, we can internalize the regular languages by means of the \emph{rational coalgebra}.

\begin{definition}
The \emph{rational coalgebra} $\rho T$ for $T$ is the colimit (taken in the category of $T$-coalgebras and homomorphisms) of all $T$-coalgebras with finitely presentable carrier.
\end{definition}

\takeout{
\begin{remark}
\begin{enumerate}[label=(\alph*)]
\item Observe that the colimit defining $\rho T$ is filtered since (i) colimits of coalgebras are formed on the level of $\D$ and (ii) $\D_f$ is closed under finite colimits.

\item In the case where the input object $X$ is finitely presentable the
  endofunctor $T$ is finitary. This implies that $\rho T$ is a
  fixpoint of $T$, see \cite{amv_atwork}. Moreover, in this case $\rho T$ is the final locally finite
  coalgebra \cite{m_linexp}, that is, every fp-coalgebra has a unique coalgebra homomorphism into $\rho T$. However, we will not use
  these properties in the present setting, so we can work with
  arbitrary input objects $X$.
\end{enumerate}
\end{remark} 
}

\begin{example}
In $\D=\Set$, the rational coalgebra is the subcoalgebra of the final coalgebra $\Pow X_0^*$ given by the set of all regular languages. Analogously for $\D=\Pos$ where the order is given by inclusion of languages. In general, $\rho T$ always has the $\D$-regular languages as states:
\end{example}

\begin{proposition}\label{prop:ratfix}
There is a one-to-one correspondence between $\D$-regular languages and elements $I\ra \rho T$ of the rational coalgebra.
\end{proposition}

\begin{proof}
We describe mutually inverse maps 
\[
(I\xra{x}\rho T)\mapsto (\xs\xra{L_x} Y) 
\qquad \text{and}\qquad 
(\xs\xra{L} Y) \mapsto (I\xra{x_L} \rho T)
\] 
between the elements of $\rho T$ and the $\D$-regular languages.
Let $h_Q: Q \to \rho T$ be the injections of the colimit $\rho T$, where $Q=(Q,\delta_Q, f_Q)$ ranges over all fp-coalgebras. Note that this colimit is filtered because $\D_f$ is closed under finite colimits in $\D$. Moreover, since colimits of coalgebras are formed in the underlying category, the morphisms $h_Q$ also form a colimit cocone in $\D$. Given an element $I\xra{x} \rho T$ of the rational coalgebra we define a $\D$-regular language $L_x: \xs \ra Y$ as follows: since $I=\Psi 1$ is finitely presentable, there exists an fp-coalgebra $Q$ and a morphism $i_Q: I\ra Q$ such that $x=h_Q\o i_Q$. For the $F$-algebra $(Q,\delta_Q, i_Q)$ we have the unique $F$-algebra homomorphism $e_{Q}: \xs \ra Q$, and we put $L_x := f_Q\o e_{Q}$. 

We need to show that $L_x$ is well-defined, that is, for any other factorization $x = h_{Q'}\o i_{Q'}$  we have $f_Q\o e_{Q}=f_{Q'}\o e_{Q'}$. Given such a factorization, since the morphisms $h_Q$ form a filtered colimit cocone, there exists an fp-coalgebra $Q''=(Q'',\delta_{Q''},f_{Q''})$ and  coalgebra homomorphisms $h_{QQ''}: Q\ra Q''$ and $h_{Q'Q''}: Q'\ra Q''$ with $h_{QQ''}\o i_Q = h_{Q'Q''}\o i_{Q'} =: i_{Q''}$. Then for the $F$-algebra $(Q'', \delta_{Q''}, i_{Q''})$ we have the unique homomorphism $e_{Q''}: \xs \ra Q''$. Moreover, $h_{QQ''}$ and $h_{Q'Q''}$ are also homomorphisms of $F$-algebras. If follows that
\begin{align*}
f_Q \o e_{Q} &= f_{Q''} \o h_{QQ''} \o e_{Q} & \text{($h_{QQ''}$ coalgebra homomorphism)}\\
&= f_{Q''} \o e_{Q''} & \text{($h_{QQ''}$ $F$-algebra hom., $\xs$ initial)}
\end{align*} 
and analogously $f_{Q'} \o e_{Q'}=f_{Q''} \o e_{Q''}$. Hence $f_Q\o e_{Q}=f_{Q'}\o e_{Q'}$, as claimed.

Conversely, let $L:\xs \ra Y$ be a $\D$-regular language. Then there exists an fp-automaton $(Q,\delta_Q,i_Q,f_Q)$ with $L = f_Q\o e_Q$. Put $x_L := h_Q\o i_Q: I \ra \rho T$. To prove the well-definedness of $x_L$, consider the automata homomorphisms 
\[ \xymatrix{
Q & Q_r \ar@{>->}[l]_-m \ar@{-->>}[r]^-e & \Min{L}
}  \]
of Theorem \ref{thm:minaut}.  Then
\begin{align*}
h_Q\o i_Q &= h_Q \o m\o i_{Q_r} & \text{($m$ algebra hom.)}\\
&= h_{Q_r} \o i_{Q_r} & \text{($(h_Q)$ cocone and $m$ coalgebra hom.)}\\
&= h_{\Min{L}} \o e \o i_{Q_r} & \text{($(h_Q)$ cocone and $e$ coalgebra hom.)}\\
&= h_{\Min{L}} \o i_{\Min{L}} & \text{($e$ algebra hom.)}
\end{align*}
Hence $x_L = h_Q\o i_Q$ is independent of the choice of $Q$. It now follows immediately from the definitions that $x\mapsto L_x$ and $L\mapsto x_L$ are mutually inverse and hence define the desired bijective correspondence.
\end{proof}

\section{Dual Characterization of Syntactic $\D$-Monoids}

We conclude this paper with a dual approach to syntactic monoids. This section is largely based on results from our papers \cite{ammu14,ammu15} where a categorical generalization of Eilenberg's variety theorem was proved. We work with the following

\begin{assumptions}
From now on $\D$ is a locally finite commutative variety of algebras or 
ordered algebras (cf. Example \ref{ex:dregassumptions}) whose epimorphisms are surjective.
Moreover, we assume that there is another locally finite variety $\C$ of 
algebras such that the full subcategories $\C_{f}$ and $\D_f$ of finite 
algebras are dually equivalent. (Two such varieties $\Cat$ and $\D$ are called 
\emph{predual}.)
\end{assumptions}

The action of the equivalence functor $\C_{f}^{op} \xra{\simeq} \D_f$ on objects and morphisms is denoted by $Q\mapsto \widehat{Q}$ and $f\mapsto \widehat{f}$. Letting $\ol{I}\in \C_f$ denote the free one-generated object of $\C$, we choose the output object $Y\in\D_f$ to be the dual object of $\ol{I}$. Moreover, let $\ol Y\in \C_f$ be the dual object of $I=\Psi 1\in \D_f$, the free one-generated object of $\D$.  Finally, we put $X = \Psi X_0$ for a finite alphabet $X_0$. Note that the underlying sets of $\ol{Y}$ and $Y$ are canonically isomorphic:
\begin{equation}\label{eq:canoniso}
\under{\ol{Y}}\cong \C(\ol{I},\ol{Y}) \cong \D(I,Y) \cong \under{Y}.
\end{equation}
To simplify the presentation, we will assume in the following that $\ol{Y}$ and $Y$ have a two-element underlying set $\{0,1\}$. This is, however, inessential; see Remark \ref{rem:twoel} at the end of this section.

\begin{example}
The categories $\C$ and $\D$ in the table below satisfy our assumptions.
\begin{center}
\begin{tabular}{|ll|}
\hline\rule[11pt]{0pt}{0pt}
$\C$ & $\D $ \\
\hline
$\BA$ & $\Set$ \\
$\DL$ & $\Pos$ \\
$\BR$ & $\PSet$\\
$\JSL$ & $\JSL$ \\
$\Vect{\Int_2}$ & $\Vect{\Int_2}$ \\
\hline
\end{tabular}
\end{center}
Here $\BA$, $\DL$ and $\BR$ are the categories of boolean algebras, 
distributive lattices with $0$ and $1$, and non-unital boolean rings (i.e.  
rings without $1$ satisfying the 
equation $x\o x= x$). The preduality of $\BA$ and $\Set$ is a restriction of 
the well-known Stone duality: the dual equivalence functor 
$\BA_f^{op}\xra{\simeq} \Set_f$ assigns to a finite boolean algebra $B$ the 
set $\At(B)$ of its atoms, and to a homomorphism $h: A\ra B$  the map $\At(h): 
\At(B)\ra\At(A)$ sending $b\in \At(B)$ to the unique atom $a\in\At(A)$ with 
$ha\geq b$. Using a similar Stone-type duality, we proved in \cite{ammu15} 
that the the category
$\BR$ of non-unital boolean rings is predual to $\PSet$. The preduality 
between $\DL$ 
and $\Pos$ is due to Birkhoff. The 
last two examples above correspond to the well-known self-duality of finite 
semilattices and finite-dimensional vector spaces, respectively. We refer to 
\cite{ammu15} for details.
\end{example}

On $\C$ we consider the endofunctor $\ol{T}Q = \ol{Y}\times Q^{X_0}$, where $Q^{X_0}$ denotes the $|X_0|$-fold power of $Q$ in $\C$. Its coalgebras are precisely deterministic automata in $\C$ without an initial state, represented as triples $(Q,\gamma_a,f)$ with transition morphisms $\gamma_a: Q\ra Q$ ($a\in X_0$) and an output morphism $f: Q\ra \ol{Y}$.

In the following we use the case $\C = \BA$ and $\D=\Set$ as our only running example. For details on the other examples see~\cite{ammu14,ammu15}
\begin{example}
	 In $\Cat=\BA$ a $\ol{T}$-coalgebra is a deterministic automaton with a boolean algebra $Q$ of states, transition maps $\gamma_a$ which are boolean homomorphisms, and an output map $f: Q\ra \{0,1\}$ that specifies (via the preimage of $1$) an ultrafilter $F\seq Q$ of final states.
\end{example}

The rational coalgebra $\rho \ol{T}$ of $\ol{T}$ (i.e.  the colimit of all finite $\ol T$-coalgebras) has as states the regular languages over $X_0$. The final state predicate $f: \rho\ol T\ra \ol Y=\{0,1\}$ sends a language to $1$ iff it contains the empty word $\epsilon$, and the transitions $\gamma_a: \rho \ol T \ra \rho\ol T$ for $a\in X_0$ are given by $\gamma_a(L)=a^{-1} L$.
Here $a^{-1}L = \{w\in X_0^*: aw\in L\}$ denotes the \emph{left derivative} of $L$ w.r.t. the letter $a$. Similarly, the \emph{right derivatives} of $L$ are defined by $La^{-1} = \{w\in X_0^*: wa\in L\}$ for $a\in X_0$. 

The coalgebra $\rho\ol T$ is characterized by a universal property: every finite $\ol T$-coalgebra has a unique coalgebra homomorphism into it (which sends a state to the language it accepts in the classical sense of automata theory). A finite $\ol T$-coalgebra is called a \emph{subcoalgebra} of $\rho \ol T$ if this unique morphism is injective, i.e.  distinct states accept distinct languages. In \cite{ammu14} we related finite $\ol T$-coalgebras in $\C$ to finite $F$-algebras in the predual category $\D$. Recall that $X=\Psi X_0$ implies $FA=I+ X\t A\cong I+\coprod_{X_0} A$, so $F$-algebras $FA\xra{\delta} A$  can be represented as triples $(A,\delta_a,i)$ with $\delta_a: A\ra A$ ($a\in X_0$) and $i: I\ra A$. They correspond to automata in $\D$ with inputs from the alphabet $X_0$ and without final states.

\begin{proposition}[see \cite{ammu14}]
\begin{enumerate}[label=(\alph*)]\item  The categories of finite $\ol{T}$-coalgebras and finite $F$-algebras are dually equivalent. The equivalence maps a finite $\ol{T}$-coalgebra $Q=(Q,\gamma_a,f)$ to its \emph{dual $F$-algebra} $\widehat Q = (\widehat Q, \widehat{\gamma_a},\widehat f)$:
\[ (\ol{Y} \xleftarrow{f} Q\xra{\gamma_a} Q) \quad\mapsto\quad (I\xra{\widehat{f}}\widehat{Q}\xleftarrow{\widehat{\gamma_a}} \widehat{Q}). \]
\item A finite $\ol{T}$-coalgebra $Q$ is a subcoalgebra of $\rho \ol{T}$ iff its dual $F$-algebra $\widehat{Q}$ is a quotient of the initial $F$-algebra $\xs$.
\end{enumerate}
\end{proposition}

\begin{example}
For a finite $\ol T$-coalgebra $(Q,\gamma_a,f)$ in $\BA$ the dual $F$-algebra $\widehat Q$ has as states the atoms of $Q$, and the initial state is the unique atomic final state of $Q$. Moreover, there is a transition $z\xra{a} z'$ for $a\in X_0$ in $\widehat Q$ iff $z'$ is the unique atom with $\gamma_a(z')\geq z$ in $Q$.	
\end{example}

 By a \emph{local variety of languages over $X_0$ in $\Cat$} we mean a subcoalgebra $V$ of $\rho \ol T$ closed under right derivatives (i.e. $L\in\under{V}$ implies $La^{-1}\in\under{V}$ for all $a\in X_0$). Note that a local variety is also closed under the $\Cat$-algebraic operations of $\rho \ol T$, being a sub\emph{algebra} of $\rho\ol T$ in $\Cat$, and under left derivatives, being a sub\emph{coalgebra} of $\rho \ol T$. 
 
 \begin{example}
 	A local variety of languages in $\BA$ is a set of regular languages over $X_0$ closed under the boolean operations (union, intersection and complement) as well as left and right derivatives. This concept was introduced by Gehrke, Grigorieff and Pin \cite{ggp08}.
\end{example}
For the following proposition recall that from every $X_0$-generated $\D$-monoid one can derive an $F$-algebra, see Definition \ref{def:asso}.

\begin{proposition}[see \cite{ammu14}]\label{prop:locvarmon}
  A finite subcoalgebra $V$ of $\rho \ol{T}$ is a local variety iff its dual $F$-algebra $\widehat{V}$ is derived from an $X_0$-generated $\D$-monoid.
\end{proposition}

In other words, given a finite local variety $V\monoto \rho \ol T$ in $\Cat$,  there exists a unique $\D$-monoid structure on $\widehat V$ making the unique (surjective) $F$-algebra homomorphism $e_{\widehat{V}}: \xs\epito \widehat{V}$ a $\D$-monoid morphism. Hence the monoid multiplication on $\widehat V$ is (well-)defined by $e_{\widehat{V}}(x)\bullet e_{\widehat{V}}(y) := e_{\widehat{V}}(x\bullet y)$ for all $x,y\in\xs$, and the unit is the initial state of the $F$-algebra $\widehat V$.

\begin{remark}
A \emph{pointed} $\ol T$-coalgebra is a $\ol T$-coalgebra $(Q,\gamma_a,f)$ equipped with an initial state $i: \ol I\ra Q$. Observe that every finite pointed $\ol T$-coalgebra $(Q,\gamma_a,f,i)$ dualizes to a finite $\D$-automaton $(\widehat Q,\widehat \gamma_a, \widehat f, \widehat i)$. The \emph{language} of $(Q,\delta_a, f,i)$ is the function
\[L_Q: X_0^*\ra \under{\ol{Y}},\quad a_1\cdots a_n \mapsto f\o\delta_{a_n}\o\dots\o \delta_{a_1}\o i,\]
where $\under{\ol Y}$ is identified with $\under{Y}$ as in \eqref{eq:canoniso}.
Letting $m_Q: Q\ra \rho \ol T$ denote the unique coalgebra homomorphism, $L_Q$ is precisely the element of $\rho \ol T$ determined by $I\xra{i} Q \xra{m_Q} \rho \ol T$.
Since $\under{\ol{Y}}=\under{Y}$ and $\xs=\Psi X_0^*$,  the function $L_Q: X_0^*\ra \under{\ol Y}$ can be identified with its adjoint transpose $L_Q^@: \xs \ra Y$, i.e.  with a language in $\D$. The \emph{reversal} of a language $L: \xs\ra Y$ in $\D$ is $L^\rev = L\o \rev: \xs\ra Y$, where $\rev: \xs \ra \xs$ denotes the unique morphism of $\D$ extending the function $X_0^*\ra X_0^*$ that reverses words.
\end{remark}

\begin{proposition}[see \cite{ammu15}]
The language accepted by a finite pointed $\ol{T}$-coalgebra is the reversal of the language accepted by its dual $\D$-automaton.
\end{proposition}

If a finite $X_0$-generated $\D$-monoid $e: \xs \ra M$ recognizes a language $L: \xs\ra Y$ via $f: M\ra Y$, i.e.  $L=(\xymatrix{\xs \ar@{->>}[r]^e& M \ar[r]^f & Y})$, we dually get the morphism $\xymatrix{\ol I \ar[r]^i & V \ar@{>->}[r]^m & \rho \ol T}$ (where $V$ is the local variety dual to $M$, $i$ is the dual morphism of $f$ and $m$ is the unique coalgebra homomorphism) choosing the element $L^\rev$ of $\rho \ol T$. Now suppose that $L$ is a regular language, and let $V_L$ be the finite local variety of languages dual to the syntactic $\D$-monoid $\Syn{L}$, see Proposition \ref{prop:locvarmon}. The universal property of $\Syn{L}$ in Definition \ref{def:syn} then dualizes as follows: (i) $V_L$ is a local variety containing $L^\rev$, and (ii) for every local variety  $V\monoto\rho \ol{T}$ containing $L^\rev$, the local variety $V_L$ is contained in $V$. In other words, $V_L$ is the \emph{smallest} local variety containing $L^\rev$.
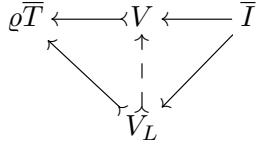
\[
\xymatrix{
  \rho \ol{T}
  &
 V \ar@{>->}[l]
  &
  \ol{I} \ar[l] \ar[dl]
  \\
  &
  V_L \ar@{>->}[ul] \ar@{>-->}[u]
}
\]
In summary, we have proved the following dual characterization of syntactic $\D$-monoids:
\begin{theorem}\label{thm:syndual}
For every regular language $L$ the syntactic $\D$-monoid $\Syn L$ is dual to the smallest local variety of languages over $X_0$ in $\C$ containing $L^\rev$.
\end{theorem}

\begin{example} For $\Cat=\BA$ and $\D=\Set$ the previous theorem gives the following construction of the syntactic monoid of a regular language $L\seq X^*$:
\begin{enumerate}[label=(\arabic*)]
	\item Form the smallest local variety of languages $V_L\seq \rho\ol T$ containing $L^\rev$. Hence $V_L$ is the closure of the (finite) set of all both-sided derivatives $u^{-1}L^\rev v^{-1} = \{w\in X^*: uwv\in L^\rev\}$ ($u,v\in X^*$)  under union, intersection and complement.
	\item Compute the $F$-algebra $\widehat{V_L}$ dual to the coalgebra $V_L$, which is a quotient $e_{\widehat{V_L}}: \xs \to \widehat{V_L}$. The states of $\widehat{V_L}$ are the atoms of $V_L$, and the initial state is the unique atom $i\in V_L$ containing the empty word. Given atoms $z,z'\in V_L$ and $a\in X$, there is a transition $z\xra{a}z'$ in $\widehat{V_L}$ iff $z'$ is the (unique) atom with $z\subseteq a^{-1}z'$.
	\item Define a monoid multiplication on $\widehat{V_L}$ as follows: given states $z,z'\in\widehat{V_L}$, choose words $w,w'\in X^*$ with $e_{\widehat{V_L}}(w) = z$ and $e_{\widehat{V_L}}(w') = z'$ (i.e.~$i\xra{w} z$ and $i\xra{w'} z'$ in $\widehat{V_L}$). Then put $z\bullet z' = e_{\widehat{V_L}}(ww')$; this is the state reached on input $ww'$, i.e.~$i\xra{ww'} z\bullet z'$. The resulting monoid (with multiplication $\bullet$ and unit $i$) is $\Syn{L}$.
\end{enumerate}
\end{example}

By dropping right derivatives and using the correspondence between finite subcoalgebras of $\rho \ol T$ and finite quotient algebras of $\xs$, one also gets the following dual characterization of minimal $\D$-automata; cf. also  \cite{bkp13} for a closely related dual prespective on minimal automata.

\begin{theorem}\label{thm:mindual}
For every regular language $L$ the minimal $\D$-automaton for $L$ is dual to the smallest subcoalgebra of $\rho \ol T$ containing $L^\rev$
\end{theorem}

\begin{remark}\label{rem:twoel}
Our assumption that $Y$ and $\ol Y$ have two elements is inessential. Without this assumption, the rational coalgebra $\rho\ol T$ is not carried by the set of regular languages, but more generally by the set of \emph{regular behaviors}, i.e, functions $b:X_0^*\ra \under{Y}$ realized by finite Moore automata with output set $\under{Y}$. The coalgebra structure is given by the output map $b\mapsto b(\epsilon)$, and transitions $b\xra{a} a^{-1}b$ for $a\in X_0$, where $a^{-1}b$ is the (generalized) left derivative of $b$ defined by $a^{-1}b(w) = b(aw)$. (Right derivatives are defined analogously.) A \emph{local variety of behaviors over $X_0$ in $\Cat$} is again a subcoalgebra of $\rho \ol T$ closed under right derivatives. All results of this section hold for this more general setting, see Section 5 of \cite{ammu15} for details. In particular, this allows us to cover the case $\C=\D=\Vect{\K}$ for an arbitrary finite field $\K$. In this case Theorem \ref{thm:syndual} states that the syntactic associative algebra of a rational weighted language $L: X_0^*\ra \K$ dualizes to the smallest set of rational weighted languages that contains $L^\rev$ and is closed under scalar multiplication, addition and left and right derivatives.
\end{remark}

\section{Conclusions and Future Work}
\label{sec:con}

We proposed the first steps of a categorical theory of algebraic language 
recognition. Despite our assumption that $\D$ be a commutative variety, we have seen that a number of our definitions, constructions and proofs work in any symmetric 
monoidal closed category with enough structure. The next step should be the treatment
of algebraic recognition beyond languages of finite words, including such 
examples as Wilke 
algebras \cite{wilke91} (representing $\omega$-languages) or forest algebras 
\cite{bw08} (representing tree languages). A first categorical approach to such structures appears in the 
work of Boj\'anczyk \cite{boj15} who works with monads on sorted sets rather than 
monoids in a variety $\D$. We expect that a unification of these two 
directions is possible. 

One of the leading themes of algebraic automata theory is the classification of languages in terms of their syntactic algebras. For instance, by Sch\"utzenberger's theorem a language is star-free iff its syntactic monoid is aperiodic. We hope that our conceptual view of syntactic monoids (notably their dual characterization in Theorem \ref{thm:syndual}) can contribute to a duality-based approach to such results.

\bibliographystyle{plain}
\bibliography{coalgebra,ourpapers}

\end{document}